\documentclass[journal,10 pt]{IEEEtran}
\usepackage{bm,cite,algorithm,algorithmic,float,amsmath,amssymb}
\usepackage{amsmath,amssymb,amsthm}
\usepackage{subfigure}

\usepackage{amssymb}
\usepackage{amsmath}
\usepackage{graphicx}
\usepackage{cite}
\usepackage{citesort}
\usepackage{balance}
\usepackage[utf8x]{inputenc}

\newtheorem{theo}{Theorem}
\newtheorem{lemma}{Lemma}
\newtheorem{cor}{Corollary}
\newtheorem{ppro}{Proposition}

\theoremstyle{remark}

\IEEEoverridecommandlockouts

\usepackage{graphicx,epstopdf}
\usepackage{epsfig}	
\usepackage{amsfonts,balance}
\usepackage{bbm}
\floatname{algorithm}{Algorithm}
\usepackage{lipsum}
\usepackage{url}

\makeatletter
\def\ScaleIfNeeded{%
\ifdim\Gin@nat@width>\linewidth \linewidth \else \Gin@nat@width
\fi } \makeatother

\usepackage{color}

\newcommand*{\mathcolor}{}
\def\mathcolor#1#{\mathcoloraux{#1}}
\newcommand*{\mathcoloraux}[3]{%
  \protect\leavevmode
  \begingroup
    \color#1{#2}#3%
  \endgroup
}

\begin{document}

\title{Artificial-Noise Aided Secure Transmission in
Large Scale Spectrum Sharing Networks}

\author{Yansha Deng,~\IEEEmembership{Student Member,~IEEE}, Lifeng Wang,~\IEEEmembership{Student Member,~IEEE}, Syed Ali Raza Zaidi,~\IEEEmembership{Member,~IEEE}, Jinhong Yuan,~\IEEEmembership{Fellow,~IEEE}, and Maged
Elkashlan,~\IEEEmembership{Member,~IEEE}
\thanks{Manuscript received May. 27, 2015; revised Nov. 21, 2015; accepted Mar. 10, 2016.  This paper was presented in part at the Proc. IEEE Int. Conf. Commun. (ICC), London, UK, Jun. 2015 \cite{yan2015ICC}. The editor coordinating the review of this
manuscript and approving it for publication was Prof. Ali Tajer.}
\thanks{Y. Deng and  L. Wang were with Queen Mary University of London, London, UK 
(e-mail:\{ y.deng, lifeng.wang\}@qmul.ac.uk).}
\thanks{ M. Elkashlan is with Queen Mary University of London, London, UK 
(e-mail:  maged.elkashlan@qmul.ac.uk).}
\thanks{Syed Ali Raza Zaidi is with the
School of Electronics and Electrical Engineering, University of Leeds,
Leeds, United Kingdom (e-mail: elsarz@leeds.ac.uk).}
\thanks{Jinhong Yuan is with the School of Electrical Engineering and
Telecommunications, The University of New South Wales, Sydney 2052,
Australia (e-mail:  j.yuan@unsw.edu.au).}
}

\maketitle
\setcounter{page}{1} \thispagestyle{plain}

\begin{abstract}
We investigate beamforming and artificial noise generation at the secondary transmitters to establish secure transmission in large scale spectrum sharing networks, where multiple non-colluding eavesdroppers attempt to intercept the secondary transmission. We develop a comprehensive analytical framework to accurately assess the secrecy performance under the primary users' quality of service constraint. Our aim is to characterize the impact of beamforming and artificial noise generation (BF$\&$AN) on this complex large scale network.  We first derive  exact expressions for the average secrecy rate and the secrecy outage probability. We then derive an easy-to-evaluate asymptotic average secrecy rate and asymptotic secrecy outage probability  when the number of  antennas at the secondary transmitter goes to infinity.  Our results show that  the equal  power allocation between the useful signal and artificial noise  is not always  the  best strategy to achieve maximum average secrecy rate in large scale spectrum sharing networks. Another interesting observation is that the advantage of BF$\&$AN over BF on the average secrecy rate is lost when the aggregate interference from the primary and secondary transmitters  is strong, such that it overtakes the effect of the  generated AN.

\end{abstract}



\section{Introduction}

The sky-rocketing growth of multimedia infotainment applications and broadband-hungry mobile devices exacerbate the stringent demand for ultra high data rate and  services. Such urgent demand is further driven by the exponential growth of smartphones, tablets, machine-to-machine (M2M) communication devices.  To cope with this,   unlicensed users are allowed to transmit on the spectrum
reserved for the wireless broadband devices as long as  the quality of service (QoS) of the primary network is satisfied \cite{goldsmith09,yan2015}. These networks are often referred to as  cognitive radio networks (CRNs). 


  The open and dynamic characteristics of CRNs have lead to several
new classes of security threats and challenges due to opportunistic utilization of licensed channels \cite{frag2013survey}. There exists six major types of attacks at the physical layer of CRNs, which are commonly known  as primary user emulation, objective function attack, learning attack, spectrum sensing data falsification, jamming attack, and eavesdropping \cite{mukherjee2014principle}. Among them, we focus on the eavesdropping attacks targeted at the secondary users (SUs). In this case, the eavesdroppers attempt to intercept and overhear the secondary transmission  without transmitting any signals.

Traditional security, which is achieved through  the higher layer cryptographic authentication and identification, becomes expensive and vulnerable to attacks.  Particularly, in the emerging large scale networks with  high  mobility terminals, the implementation and management of higher-layer key distribution  face increasing challenges    \cite{wang2011anti,zhou2011throughput}. In other words, the establishment of  the secret keys to achieve encrypted transmission in large scale decentralized networks are even more complicated  and expensive than the point-to-point  communications \cite{zhou2011throughput}. To cope with these issues,  physical layer security has been proposed as a complementary security method to protect the confidential information from eavesdropping \cite{wyner1975wire}.
A comprehensive overview of physical layer security in multiuser wireless networks has been presented in \cite{mukherjee2014principle}.

  Recently, physical layer security has been introduced into large scale wireless networks with randomly located eavesdroppers \cite{yan2016sensor,tong2014transmission,zheng2015multi},  single antenna legitimate nodes and eavesdroppers \cite{zhou2011throughput},  multiple jammers \cite{chao2015uncoordinated}, and  cellular users  \cite{he2013physical,Geraci2014}. In these works,  the  stochastic geometry and random  graphs were applied for modeling these networks \cite{pinto2012secure1}. This mathematical tool is attractive since it captures the topological randomness of these networks, and provides a simple and tractable model for characterizing the performance \cite{elsawy2013stochastic}.

 Various advanced  techniques have been developed to  enhance the secrecy performance  \cite{huiming2015}.
Beamforming (BF) is proved to be the optimal transmission scheme to achieve the maximum achievable  secrecy rate in multiple input single output (MISO) systems \cite{shafiee2007achievable}.
 Generating   artificial noise (AN) at the legitimate transmitter  is proposed to be an effective technique to confound the eavesdroppers  \cite{goel2008guaranteeing}.  In the AN-based method, the power allocation between the information-bearing signal and the AN at the transmitter is critically important, which reveals the tradeoff between enhancing the main channel by increasing the power allocated to information-bearing signal and degrading the eavesdropper's channel by allocating more power to the AN. In \cite{xiang2010secure} and \cite{xi2013}, the optimal power allocation strategies  were studied in the conventional wireless network with fixed nodes, and wireless ad hoc networks with mobile nodes, respectively. However, all these works \cite{shafiee2007achievable,swundke2009fixed,goel2008guaranteeing,xiang2010secure,xi2013} have considered the physical layer security   in legacy  networks.

Compared with the physical layer security in  conventional networks, there exist several major differences in the security of CRNs:
1) the QoS requirement of the primary network needs to be satisfied;
2) the SU receiver is subject to the aggregate interference from the PU transmitters;
and
3) the secondary network is more susceptible to security threats.     In light of the aforementioned circumstances, the research on enhancing the security at the physical layer of CRNs has received increasing attentions.
In \cite{stanojev2013improve,ning2013cooperative,yuan2014secrecy},
the secondary user  acts as a jammer to enhance the secrecy transmission of a primary network.
In \cite{yulong2013pysical},  multiuser scheduling was proposed to improve the security level of  secondary transmission
with  primary QoS constraint. In \cite{yu2014secrecy}, it  was demonstrated that the best secrecy performance of secondary network can be achieved when the perfect channel state information (CSI) of all links
are available. In \cite{yiyang2013secure}, it was proved that beamforming is the optimal transmission strategy to secure MISO CRN with the perfect knowledge of all channels. The authors of \cite{yiyang2013secure} then extended their work to the  networks with imperfect knowledge  in \cite{yiyang11secure}. 
In \cite{chao2014secrecy},
the beamforming and artificial noise generation (BF$\&$AN) was adopted at the SU transmitter to enhance the secrecy throughput of a multiple-input, single-output, multieavesdropper (MISOME) primary network. 
Note that \cite{yulong2013pysical,yu2014secrecy,yiyang2013secure,yiyang11secure,chao2014secrecy} only considered  fixed location nodes. In \cite{anand2008cognitive}, the secrecy capacity of cognitive radio networks with uniformly distributed secondary transmitters and primary transmitters was examined.  In \cite{shu2013physical}, the secrecy capacity of the primary network was analyzed in CRNs, where  PUs, SUs, and Eves followed the mutually independent homogeneous Poisson process.

Different from the aforementioned works, we treat the secrecy performance of large scale spectrum sharing networks with  BF$\&$AN at the SU transmitters.
 Compared against the security of non-cognitive radio networks in \cite{xi2013},  the prerequisite of underlay spectrum sharing networks is to guarantee the QoS of the primary network. This can be  fulfilled by constraining the outage probability at the PU receiver below a predetermined threshold (i.e.,  the peak allowable outage probability).
The use of BF$\&$AN at the SU transmitter    brings  array gains at the legitimate receiver and disrupts the reception at the eavesdropper.
 Although  BF$\&$AN  has been well treated in the conventional physical layer security network  in \cite{xiang2010secure}, no work has  considered  BF$\&$AN    in large scale   spectrum sharing  networks. Therefore, the question of \emph{how  BF$\&$AN impacts the security design of such a complex network} remains unknown.   Our contributions are summarized as follows:
 \begin{enumerate}
\item  We derive a new exact closed-form expression for the maximum permissive transmit power at the SU transmitter with BF$\&$AN. We  accurately quantify the permissive transmit power region where the primary network's QoS can be guaranteed, as presented in {\bf{Theorem 1}}.
 We derive the exact  expressions for the average secrecy rate and the secrecy outage probability of the secondary network with BF$\&$AN at the SU transmitters, as presented in {\bf{Theorems 2}} and {\bf{3}}.
\item We show that there exists an average secrecy rate boundary beyond which the PU receiver's QoS is violated.
 We reveal that the optimal power allocation factor   for maximizing the average secrecy rate varies for  different system parameters.   Equal power allocation  may not achieve the near optimal average secrecy rate in  large scale spectrum sharing networks.
\item
To provide insights into  system design from  an implementation viewpoint, we compare the average secrecy rate  of BF$\&$AN   with that of BF.  We observe  the same average secrecy rate boundary  for  BF$\&$AN and BF. The advantage of BF$\&$AN  over BF  on the average secrecy rate  is lost, when the aggregate interference from the PU and SU transmitters  is strong, such that it overtakes the effect of the  generated AN.
\item  We derive the asymptotic average secrecy rate and the asymptotic secrecy outage probability of the secondary network with BF$\&$AN at the SU transmitters   when  the number of SU transmit antennas $N_s$ goes to infinity, as presented in {\bf{Propositions 1, 2}}, and {\bf{3}}. Our asymptotic results well predict the  exact performance in the medium and large $N_s$ regime.
We determine the  antenna gap, which showcases the number of additional antennas  required to achieve the  same asymptotic average secrecy rate in more dense networks.
\end{enumerate}

\begin{figure}[t!]
    \begin{center}
        \includegraphics[width=2.5in]{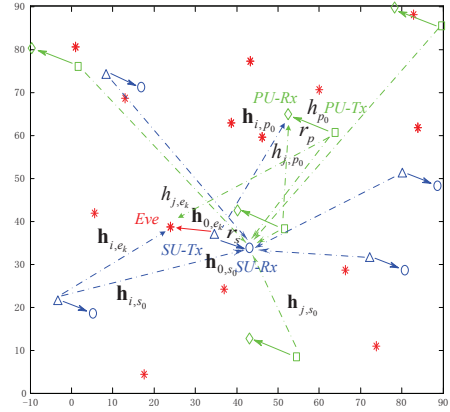}
        \caption{A realization of a large scale spectrum sharing network model describing the received signal  at a SU receiver.   In this network,  the green square represents the PU transmitter, the diamond represents the PU receiver, the triangle represents the SU transmitter, the circle represents the SU receiver, and the red star represents the eavesdropper. The blue solid line represents the secondary transmission, the green solid line represents the primary transmission, the blue dashed line represents the interference from the SU transmitter, and the green dashed line presents the interference from the PU transmitter.}
        \label{fig:11}
    \end{center}
\end{figure}


\section{System and Channel Model}
\label{sec:system}
We consider  secure communication in an underlay spectrum sharing network where the SU transmitters communicate with the corresponding SU receivers under the potential malicious attempt of multiple eavesdroppers. Each SU transmitter has $N_s$ antennas, and the remaining nodes in this model are all  single-antenna nodes.
As shown in Fig. \ref{fig:11},
we have a set of PU transmitters, SU transmitters, and eavesdroppers locations, denoted by $\Phi_p$, $\Phi_s$ and $\Phi _e$, in which $\Phi_p$, $\Phi_s$ and $\Phi _e$ follow independent homogeneous Poisson point processes (HPPPs) with densities $\lambda_p$, $\lambda_s$ and $\lambda _e$, respectively. This model is practical and representative of the decentralized networks, where the nodes are randomly deployed or have substantial mobility \cite{weber2010overview}. 
We assume that  each PU/SU transmitter communicates with its unique associated intended PU/SU receiver at a fixed distance,  respectively, in order to simplify the analysis and provide some design insights \cite{xi2013,zaidi2014,baccelli2006}. Note that this fixed distance assumption can be relaxed by taking into account the probability density distribution of the distance.



 The wireless channels are modeled as independent quasi-static Rayleigh fading.
  The eavesdroppers  interpret the secondary transmitter's signal without trying to modify it.
  In this complex CRNs, we consider the interference-limited case where the thermal noise is negligible compared with the aggregate interference from the other transmitters. Similar as \cite{xiang2010secure,xi2013}, we utilize the SIR to characterize the performance.


%
We mask the beamformed broadcast information with the AN at the SU transmitters to confuse the eavesdroppers. Each SU transmitter  broadcasts the information-bearing signals  and AN simultaneously. 
We assume that the perfect CSI between each SU transmitter and each SU receiver are available\footnote{In practice, perfect CSI may not be easy to obtained, as such, our analysis
provides the upper bound on the actual achievable secrecy performance.}.
The AN is transmitted in the null space of the intended SU receiver's channel, thus imposing no effect on the secondary channel, whereas degrading the eavesdropper's channel. We  denote the intended channel vector between the $i$th SU transmitter $\left({i} \in {\Phi _s}\right)$ and the corresponding SU receiver as $\bold{h}_{i,s_i}\in \mathcal{C}^{1 \times {N_s}}$, the channel state information (CSI) of which is known at the $i$th SU transmitter. An orthonormal basis of $\mathcal{C}^{{N_s}\times {N_s}}$ is generated at the $i$th SU transmitter   as $ {\big[ {{{{{\bf{h}}_{i,{s_i}}^\dag} } \mathord{\left/
 {\vphantom {{{{\bf{h}}_{i,{s_i}}}^\dag } {\left\| {{{\bf{h}}_{i,{s_i}}}} \right\|}}} \right.
 \kern-\nulldelimiterspace} {\left\| {{{\bf{h}}_{i,{s_i}}}} \right\|}}{\rm{ }},\:{{\bf{G}}_{i,{s_i}}}} \big]_{{N_s} \times {N_s}}}$ \footnote{$^\dag$ is the conjugate transpose operator.}, where ${\bold{G}_{i,s_i}}$ is a ${{N_s} \times ({N_s-1})}$ matrix. Note that each column of ${\bold{G}_{i,s_i}}$ and ${{{{{\bf{h}}_{i,{s_i}}^\dag} } \mathord{\left/
 {\vphantom {{{{\bf{h}}_{i,{s_i}	}}^\dag } {\left\| {{{\bf{h}}_{i,{s_i}}}} \right\|}}} \right.
 \kern-\nulldelimiterspace} {\left\| {{{\bf{h}}_{i,{s_i}}}} \right\|}}}$ are mutually orthogonal. We define $b_i$ as the information-bearing signal, and  $\bold{n}_A$ as the AN.
 The transmitted BF$\&$AN symbol vector is modeled as
 \begin{align}\label{transmit signal}
{{\bf{x}}_{s_i}} = \frac{{{{\bf{h}}_{i,s_i}^\dag}}}{{\left\| {{{\bf{h}}_{i,s_i}}} \right\|}}b_i + {{\bf{G}}_{i,s_i}}{\bf{n}}_A,
\end{align}
   where   $\mathbbm{E}\left\{ {{b_i}b_i^\dag } \right\} = \delta _s^2$, and $N_s-1$ elements of  $\bold{n}_A$ are independent and identically distributed (i.i.d) complex Gaussian random variables  with zero mean and variance  $\sigma _n^2$. Thus, the total transmit power per transmission $P_s$ is given by ${P_s} = {P_I} + {P_A}$, where the power allocated to the information signal is  ${P_I} = \sigma _s^2$ and the power allocated to the AN is ${P_A} = \left( {{N_s} - 1} \right)\sigma _n^2$. We also define $\mu $ as  the fraction of power assigned to the information signal, thus $ P_I = \mu {P_s}$ .


 In the primary network, we assume the typical PU receiver is located at the  origin of the coordinate system, and the distance between the typical PU transmitter and its associated PU receiver is $r_{p}$. According to the Slivnyak's theorem \cite{stoyanstochastic}, adding a probe point
to the HPPP at an arbitrary location does not affect the law
of the point process.  The received SIR at the typical PU receiver is given by
\begin{align}\label{PU_SIR_AN}
\gamma _{\tt SIR}^{p,AN} = \frac{{{{\left| {{h_{{p_0}}}} \right|}^2}{r_p}^{ - \alpha }}}{{{I_{p,{p_0}}} + P_p^{ - 1}{I_{s,{p_0}}}}},
\end{align}
where   
\begin{align}
{I_{p,{p_0}}} = {\sum _{j \in {\Phi _p}\backslash \left\{ 0 \right\}}}{\left| {{h_{j,{p_0}}}} \right|^2}{\left| {{X_{j,{p_0}}}} \right|^{ - \alpha }},
\end{align} 
and 
 \begin{align}
{I_{s,{p_0}}} = {\sum _{i \in {\Phi _s}}}\Big[ {\sigma _s^2{{\big| {{{\bf{h}}_{i,{p_0}}}\frac{{{\bf{h}}_{i,{s_i}}^\dag }}{{\left\| {{\bf{h}}_{i,{s_i}}^\dag } \right\|}}} \big|^2}} + \sigma _n^2{{\big\| {{{\bf{h}}_{i,{p_0}}}{{\bf{G}}_{i,{s_i}}}} \big\|}^2}} \Big]{\big| {{X_{i,{p_0}}}} \big|^{ - \alpha }}.
\end{align} 

In \eqref{PU_SIR_AN}, $\alpha$ is the path-loss exponent, ${{h_{{p_0}}}}$ is the channel fading gain  between the typical PU transmitter and the typical PU receiver,  ${{h}_{j,{p_0}}} $  and ${\left| {{X_{j,{p_0}}}} \right|}$ are the  interfering channel fading gain and distance between the $j$th PU transmitter and the typical PU receiver, respectively.  ${\bold{h}_{i,{p_0}}} \in \mathcal{C}^{1 \times {N_s}}$ and  ${\left| {{X_{i,{p_0}}}} \right|}$ are   the  interfering channel vector and distance between the $i$th SU transmitter and the typical PU receiver, respectively. ${P_p}$ is the transmit power at the PU  transmitter. Note that
${P_p}{I_{p,{p_0}}}$ is the interference from other PU  transmitters to the typical PU receiver, ${I_{s,{p_0}}}$ is the co-channel interference from the SU  transmitters to the typical PU receiver.


In the secondary network, we    assume ${{\bf{h}}_{{0,s_0}}}\in \mathcal{C}^{1 \times {N_s}}$ and $r_{s}$ to be the  channel vector and
  distance between the typical SU  transmitter and corresponding typical SU receiver.
 Note that  each SU transmitter transmits the signal vector  expressed as  \eqref{transmit signal}, we obtain the effective signal at the typical SU  receiver as 
{
\begin{align}
 {{\bf{h}}_{{0,s_0}}}{{\bf{x}}_{{s_0}}} = {{\bf{h}}_{{0,s_0}}}\frac{{{{\bf{h}}_{{0,s_0}}^\dag} }}{{\left\| {{{\bf{h}}_{{0,s_0}}}} \right\|}}{b_0} + {{\bf{h}}_{{0,s_0}}}{{\bf{G}}_{{0,s_0}}}{{\bf{n}}_A} = {\left\| {{{\bf{h}}_{0,{s_0}}}} \right\|}{b_0}.
 \end{align}}
  The received SIR at the typical SU receiver is given by
\begin{align}\label{SU_SIR_AN}
\gamma _{\tt SIR}^{s,AN} = \frac{{\sigma _s^2{{\left\| {{{\bf{h}}_{0,{s_0}}}} \right\|}^2}{r_s}^{ - \alpha }}}{{{I_{s,{s_0}}} + {P_p}{I_{p,{s_0}}}}},
\end{align}

where  
\begin{align}
{I_{p,{s_0}}} = {\sum _{j \in {\Phi _p}}}{\left| {{h_{j,{s_0}}}} \right|^2}{\left| {{X_{j,{s_0}}}} \right|^{ - \alpha }},
\end{align}  
and 
{\begin{align}
{I_{s,{s_0}}} = &{\sum _{i \in {\Phi _s}\backslash \left\{ 0 \right\}}}\Big[ {\sigma _s^2{{\big| {{{\bf{h}}_{i,{s_0}}}\frac{{{\bf{h}}_{i,{s_i}}^\dag }}{{\left\| {{\bf{h}}_{i,{s_i}}^\dag } \right\|}}} \big|^2}} + \sigma _n^2{{\big\| {{{\bf{h}}_{i,{s_0}}}{{\bf{G}}_{i,{s_i}}}} \big\|}^2}} \Big]
\nonumber \\&
{\big| {{X_{i,{s_0}}}} \big|^{ - \alpha }}.
\end{align}}
 In \eqref{SU_SIR_AN},   ${{h}_{j,{s_0}}}  $  and ${\left| {{X_{j,{s_0}}}} \right|}$ are the channel fading gain and distance between the $j$th PU transmitter and the typical SU receiver, respectively. ${\bold{h}_{i,{s_0}}} \in \mathcal{C}^{1 \times {N_s}}$ and  ${\left| {{X_{i,{s_0}}}} \right|}$ are   the  interfering channel vector and distance between the $i$th SU transmitter and the typical SU receiver, respectively.
Note that ${P_p}{I_{p,{s_0}}} $ is the co-channel interference from  the PU transmitters to the typical SU receiver, and ${I_{s,{s_0}}} $ is the aggregate interference from other SU transmitters  to the typical SU receiver.

In the eavesdropping channel, we consider the most detrimental eavesdropper  that has the highest SIR for a typical SU transmitter\cite{LunDong}. Note that eavesdroppers are only interested in the secondary transmissions, and interpret the primary transmissions as interference\footnote{This assumption is  practical since the primary networks operate in the Digital Video Broadcasting
(DVB) spectrum and broadcast the public service to  households, which do not have any confidential messages.}.
We  assume  ${{\bf{h}}_{{0},e_k}}\in \mathcal{C}^{1 \times {N_s}}$  to be the  channel vector between the typical SU  transmitter and an arbitrary eavesdropper $e_k \in {\Phi _e}$.
With  BF$\&$AN at the SU transmitter, the received signal from the typical SU transmitter at the $k$th eavesdropper
   is given by 
{\begin{align}   
   {{\bf{h}}_{{0},{e_k}}}{{\bf{x}}_{{s_0}}} = {{\bf{h}}_{{0},{e_k}}}\frac{{{{\bf{h}}_{{0,s_0}}^\dag} }}{{\left\| {{{\bf{h}}_{{0,s_0}}}} \right\|}}{b_0} + {{\bf{h}}_{{0},{e_k}}}{{\bf{G}}_{{0,s_0}}}{{\bf{n}}_A},
\end{align} }  
where the first part is the useful received information signal, and the second part is the received  AN.  As such,
the  SIR at  the most detrimental eavesdropper  is expressed as
\begin{align}\label{Eve_SIR_AN}
&\gamma _{\tt SIR}^{e,AN} = \mathop {\max }\limits_{{{e_k} \in {\Phi _e}}} \left\{ \frac{{\sigma _s^2{{\big| {{{\bf{h}}_{{0},{e_k}}}\frac{{{\bf{h}}_{{0,s_0}}^\dag }}{{\left\| {{\bf{h}}_{{0,s_0}}^\dag } \right\|}}} \big|^2}}{{\left| {{X_{e_k}}} \right|}^{ - \alpha }}}}{{{I_{s,{e_k}}} + {P_p}{I_{p,{e_k}}} + \sigma _n^2{I_{{s_0},{e_k},an}}}} \right\},
\end{align}
where   
\begin{align}
{I_{p,{e_k}}} = {\sum _{j \in {\Phi _p}}}{\left| {{h_{j,{e_k}}}} \right|^2}{\left| {{X_{j,{e_k}}}} \right|^{ - \alpha }},
\end{align}
\begin{align}
{I_{s,{e_k}}} = &{\sum _{i \in {\Phi _s}\backslash \left\{ 0 \right\}}}\Big[ {\sigma _s^2{{\big| {{{\bf{h}}_{i,{e_k}}}\frac{{{\bf{h}}_{i,{s_i}}^\dag }}{{\left\| {{\bf{h}}_{i,{s_i}}^\dag } \right\|}}} \big|^2}} + \sigma _n^2{{\big\| {{{\bf{h}}_{i,{e_k}}}{{\bf{G}}_{i,{s_i}}}} \big\|}^2}} \Big] \nonumber
\\ & {\left| {{X_{i,{e_k}}}} \right|^{ - \alpha }},
\end{align}
 and 
\begin{align}
{I_{{s_0},{e_k},an}} = {\left\| {{{\bf{h}}_{{0},{e_k}}}{{\bf{G}}_{{0,s_0}}}} \right\|^2}{\left| {{X_{e_k}}} \right|^{ - \alpha }}.
\end{align}
 Note that ${{h}_{j,{e_k}}}$ and ${\left| {{X_{j,{e_k}}}} \right|}$ are the channel fading gain and distance between the $j$th PU transmitter and the $k$th eavesdropper, respectively.
 ${\bold{h}_{i,{e_k}}} \in \mathcal{C}^{1 \times {N_s}}$ and ${\left| {{X_{i,{e_k}}}} \right|}$ are the channel vector and distance between the $i$th SU transmitter and the $k$th eavesdropper, respectively. ${\left| {{X_{{e_k}}}} \right|}$ is the distance between the typical SU transmitter and the $k$th eavesdropper.
It is known that ${P_p}{I_{p,{e_k}}}$ is the aggregate interference from PU  transmitters,  $\sigma _n^2{I_{{s_0},{e_k},an}}$ is the AN from the typical SU  transmitter, and ${I_{s,{e_k}}}$ is the aggregate interference from other SU transmitters.

We now define 
\begin{align}
{W_{{s_i},z}} = \sigma _s^2{\big| {{{\bf{h}}_{i,z}}\frac{{{\bf{h}}_{i,{s_i}}^\dag }}{{\left\| {{\bf{h}}_{i,{s_i}}^\dag } \right\|}}} \big|^2} + \sigma _n^2{\big\| {{{\bf{h}}_{i,z}}{{\bf{G}}_{i,{s_i}}}} \big\|^2},
\end{align}
 where $\bold{h}_{i,s_i}$ is the intended channel, $\bold{h}_{i,z}$
is the channel between the $i$th SU transmitter and  the non-intended receiver $z$ (except for  the $i$th SU receiver), and $z \in \left\{ {{p_0},{d_{0,}}{e_k}} \right\}$.
To facilitate the performance analysis,  we derive the  Laplace transform  of the aggregate interference from the SU  transmitters ${I_{s,z}} = \sum\nolimits_{i \in {\Phi _s}} {{W_{{s_i},z}}{{\left| {{X_{i,z}}} \right|}^{ - \alpha }}} $  in \eqref{PU_SIR_AN}, \eqref{SU_SIR_AN}, and \eqref{Eve_SIR_AN} as the following lemma.
\begin{lemma}
The Laplace transform of the   interference from  the SU  transmitters with BF$\&$AN to the non-intended receiver ${{{I_{s,z}}}}$ is derived as
\begin{align}\label{L_PDF_Wz}
&{\mathcal{L}_{{I_{s,z}}}}\left( s \right)  =
\left\{
\begin{array}{ll}
 \exp \Big( { - {\lambda _s}\pi P_s^{\frac{2}{\alpha }}{\Upsilon _1}\Gamma \left( {1 - \frac{2}{\alpha }} \right){s^{\frac{2}{\alpha }}}} \Big)
 \hspace{0.8cm}\mbox{$\mu  \ne \frac{1}{{{N_s}}}$},\\
\exp \Big( { - {\lambda _s}\pi {{\left( {\mu {P_s}} \right)}^{\frac{2}{\alpha }}}\Gamma \left( {{N_s} + \frac{2}{\alpha }} \right)\frac{{\Gamma \left( {1 - \frac{2}{\alpha }} \right)}}{{\Gamma \left( {{N_s}} \right)}}{s^{\frac{2}{\alpha }}}{\rm{ }}} \Big)\\
\hspace{5.6cm}\mbox{$\mu  = \frac{1}{{{N_s}}}$},
\end{array}
\right.
\end{align}
where
\begin{align}\label{Upsilon_1}
&{\Upsilon _1} ={\Big( {1 - \frac{{\left( {1 - \mu } \right)}}{{\left( {{N_s} - 1} \right)\mu }}} \Big)^{1 - {N_s}}}\Big[ {{{ \mu  }^{\frac{2}{\alpha }}}\Gamma \big( {1 + \frac{2}{\alpha }} \big)}-\frac{1}{\mu } \Big.
\nonumber \\
 &\Big. {{{\Big( {\frac{{\left( {1 - \mu } \right)}}{{{N_s} - 1}}} \Big)}^{1 + \frac{2}{\alpha }}}\sum\limits_{k = 0}^{{N_s} - 2} {{{\Big( {1 - \frac{{\left( {1 - \mu } \right)}}{{\left( {{N_s} - 1} \right)\mu }}} \Big)}^k}\frac{{\Gamma \left( {k + 1 + \frac{2}{\alpha }} \right)}}{{\Gamma \left( {k + 1} \right)}}} } \Big].
\end{align}
\end{lemma}

\begin{proof}
See Appendix A.
\end{proof}

\bigskip

\section{Exact Secrecy Perfromance}
\label{exacr secrecy}
In this section, we  first present the SU's permissive transmit power  region.  We then present the exact expressions for the average secrecy rate and the secrecy outage probability in large scale  spectrum sharing networks with BF$\&$AN at the SU transmitters. To obtain key insights through a comparison of BF$\&$AN with BF, we  derive  exact expressions for the  average secrecy rate and the secrecy outage probability in large scale  spectrum sharing networks  with BF at the SU transmitters.

\subsection{Beamforming and Artificial Noise Generation}

\subsubsection{PUs' Quality of Service Requirement}According to the rule of  underlay spectrum sharing networks, the concurrent transmission of PUs and SUs occurs under the prerequisite that   the QoS requirement of the   primary transmission  is satisfied \cite{yan2014GSC}.
As such, we first examine the transmit power operating region at the SU transmitters under  the primary network's QoS constraint.
The  QoS of primary network  is characterized that the outage probability should be no larger than the peak allowable value $\rho^p_{\tt out}$, which is expressed as \cite{yan2015GSC}
\begin{align}\label{PU_outageP}
{ P}^{\left\{ p \right\}}_{ out}={Pr}\big\{ \gamma^{p,AN}_{\tt SIR} < {{\gamma _{th}^{\left\{ p \right\}} }}\big \} < \rho^{\left\{ p \right\}}_{ out},
\end{align}
where $\gamma^{\left\{ p \right\}}_{ th}$ is the desired SIR threshold at the PU receiver.

In the following theorem, we present the SU's permissive transmit power region.
\begin{theo}
With BF$\&$AN at the SU transmitter,  the permissive transmit power region at the SU transmitter  is given as
${P_s} \in \left( {0,} \right.\left. {P_s^{\max }} \right]$, where
\begin{align}\label{Ps_AN}
&{P_s^{\max }}  =\left\{
\begin{array}{ll}
{\Big( { - \frac{\Theta }{{{\Upsilon _1}{\lambda _s}}}} \Big)^{\frac{\alpha }{2}}}{P_p}
 \hspace{1.6cm}\mbox{$\mu  \ne \frac{1}{{{N_s}}}$}\\
{\Big( { - \frac{{\Theta \Gamma \left( {{N_s}} \right)}}{{{\lambda _s}\Gamma \left( {{N_s} + \frac{2}{\alpha }} \right)}}} \Big)^{\frac{\alpha }{2}}}\frac{{{P_p}}}{\mu }
\hspace{0.6cm}\mbox{$\mu  = \frac{1}{{{N_s}}}$},
\end{array}
\right.
\end{align}
where $\Upsilon_1$ is given by \eqref{Upsilon_1},
and
\begin{align}\label{theta}
&\Theta  = \frac{{\ln \big( {1 - \rho^{\left\{ p \right\}}_{ out} \big)} \big)}}{{\pi \Gamma \left( {1 - \frac{2}{\alpha }} \right){{\big( {\gamma _{th}^{\left\{ p \right\}}} \big)}^{\frac{2}{\alpha }}}{r_p}^2}} + {\lambda _p}\Gamma \big( {1 + \frac{2}{\alpha }} \big).
\end{align}
\end{theo}
\begin{proof}
See Appendix B.
\end{proof}
The following  are some observations from \eqref{Ps_AN}.
\begin{itemize}
\item  For the fixed primary network's QoS constraint,  the maximum permissive transmit power at the SU transmitter can be relaxed by reducing the distance of the typical PU transceivers $r_p$, due to the fact that the PU can tolerate more interference from the SU transmitters.
\item With  increasing number of SU nodes and PU nodes per unit area, the transmit power constraint imposed on the SU transmitter is more severe.
This is due to the increasing aggregate interference from the SU transmitters and the other interfering PU transmitters.
\end{itemize}

\bigskip

To  study the impact of BF$\&$AN on the secrecy performance  within the permissive transmit power region, we  consider two important metrics:
the average secrecy rate and the secrecy outage probability.

\subsubsection{Average Secrecy Rate}

\bigskip

The instantaneous secrecy rate is defined as~\cite{LunDong}
\begin{align} \label{eq:CS}
R_{\tt se} =[\log_{2}\big(1+\gamma^{s,AN}_{\tt SIR} \big)-\log_{2}\big(1+\gamma _{\tt SIR}^{e,AN} \big)]^+.
\end{align}
where $[x]^+=\mathrm{max}\{x,0\}$. Here, $\gamma _{\tt SIR}^{e,AN} = \mathop {\max }\limits_{{e_k \in {\Phi _e}}} \big\{ {\gamma _{\tt SIR}^{{e_k},AN}} \big\}$  corresponds to the non-colluding eavesdropping case~\cite{Xiangyun2010}.

The average secrecy rate is the average of the instantaneous secrecy rate $R_{\tt se}$ over $\gamma _{\tt SIR}^{s,AN}$ and $\gamma _{\tt SIR}^{e,AN}$. As such, the average secrecy rate is given by \cite{lifeng2014physical}
\begin{align}\label{y_31}
{{\bar R}_{\tt se}} &= \int_0^\infty  {\int_0^\infty  {{R_{\tt se}}{f_{{\gamma^{s,AN}_{\tt SIR}}}}\left( {{x_1}} \right){f_{{\gamma _{\tt SIR}^{e,AN}}}}\left( {{x_2}} \right)} } d{x_1}d{x_2}\nonumber\\
&= \frac{1}{{\ln 2}}\int_0^\infty  {\frac{{{F_{{\gamma _{\tt SIR}^{e,AN}}}}\left( {{x_2}} \right)}}{{1 + {x_2}}}\big( {1 - {F_{{\gamma^{s,AN}_{\tt SIR}}}}\big( {{x_2}} \big)} \big)} d{x_2}.
\end{align}

In order to examine the average secrecy rate, we  derive the CDFs of SIRs  at the typical SU receiver and the most detrimental eavesdropper in the following {{Lemma 2}} and {{Lemma 3}}, respectively.

\begin{lemma}
With BF$\&$AN at the SU transmitters,  the CDF of SIR at the typical SU receiver is derived as
\begin{align}\label{CDFsecan2}
&F_{\gamma _{\tt SIR}^{s,AN}}^{\left\{ s \right\}}\big( {\gamma _{th}^{\left\{ s \right\}}} \big) =
1 - \exp \big( { - {\Lambda _l}{{\big( {\gamma _{th}^{\left\{ s \right\}}} \big)}^{\frac{2}{\alpha }}}r_s^2} \big) - \sum\limits_{m = 1}^{{N_s} - 1} {\frac{{{{\left( {r_s^\alpha } \right)}^m}}}{{m!{{\left( { - 1} \right)}^m}}}}
\nonumber \\
& \hspace{1.0cm} \sum \frac{{m!}}{{\prod\limits_{i = 1}^m {{m_i}!i{!^{{m_i}}}} }}\exp \big( { - {\Lambda _l}{{\big( {\gamma _{th}^{\left\{ s \right\}}} \big)}^{\frac{2}{\alpha }}}r_s^2} \big)
\nonumber \\
& \hspace{1.0cm} \prod\limits_{j = 1}^m {{{\Big( {\big( { - {\Lambda _l}{{\big( {\gamma _{th}^{\left\{ s \right\}}} \big)}^{\frac{2}{\alpha }}}} \big){{\left( {{r_s}} \right)}^{2 - j\alpha }}\prod\limits_{k = 0}^{j - 1} {\big( {\frac{2}{\alpha } - k} \big)} } \Big)}^{{m_j}}}},
\end{align}
where
\begin{align}\label{lambda}
&\Lambda _l = \left\{
\begin{array}{ll}
 \Lambda _2
 \hspace{0.5cm}\mbox{$\mu  = \frac{1}{{{N_s}}}$}\\
\Lambda _3
\hspace{0.5cm}\mbox{$\mu  \ne \frac{1}{{{N_s}}}$}.
\end{array}
\right.
\end{align}
In \eqref{lambda}, $\Lambda_2$ and $\Lambda_3$ are given by
\begin{align}\label{lambda2}
{\Lambda _2} = &\pi \Big( {{\lambda _s}\frac{{\Gamma \big( {{N_s} + \frac{2}{\alpha }} \big)}}{{\Gamma \big( {{N_s}} \big)}} + {\lambda _p}\Gamma \big( {1 + \frac{2}{\alpha }} \big){{\big( {\mu {\frac{{{P_s}}}{{{P_p}}}} } \big)}^{ - \frac{2}{\alpha }}}} \Big)
\Gamma \big( {1 - \frac{2}{\alpha }} \big) ,
\end{align}
\begin{align}\label{lambda3}
{\Lambda _3} =  \pi \Big( {{\lambda _p}\Gamma \big( {1 + \frac{2}{\alpha }} \big){\big({\frac{{{P_s}}}{{{P_p}}}}\big) ^{ - \frac{2}{\alpha }}} + {\lambda _s}{\Upsilon _1}} \Big)\Gamma \big( {1 - \frac{2}{\alpha }} \big){\left( \mu  \right)^{ - \frac{2}{\alpha }}},
\end{align}
respectively. Here, $\sum\limits_{i = 1}^m {i \cdot {m_i}}  = m$, and $\Upsilon_1$ is given by \eqref{Upsilon_1}, and $P_s$ is the maximum permissive transmit power, which is  given in \eqref{Ps_AN}.
\end{lemma}

\begin{proof}
See Appendix C.
\end{proof}


Based on the SIR at the most detrimental eavesdropper in \eqref{Eve_SIR_AN}, we derive the CDF for  ${\gamma _{\tt  SIR}^{e,AN}}$ in the following lemma.
\begin{lemma}
With BF$\&$AN at the SU transmitters,  the CDF of SIR at the most detrimental eavesdropper   is derived as
\begin{align}\label{CDFeveAN_2}
&F_{\gamma _{\tt SIR}^{{e},AN}}^{\left\{ e \right\}}\big( {\gamma _{th}^{\{ e\} }} \big){\rm{ = }}\nonumber \\
&\hspace{0.5cm} \exp \Big( { - \frac{{\pi {\lambda _e}}}{{{\Lambda _l}}}{{\big( {\gamma _{th}^{\{ e\} }} \big)}^{ - \frac{2}{\alpha }}}{{\big( {\frac{{1 - \mu }}{{\big( {{N_s} - 1} \big)\mu }}\gamma _{th}^{\{ e\} } + 1} \big)}^{1 - {N_s}}}} \Big),
\end{align}
where $\Lambda _l$ is given in \eqref{lambda}. Note that $P_s$ is the maximum permissive transmit power, which is  given in \eqref{Ps_AN}.
\end{lemma}
\begin{proof}
See Appendix D.
\end{proof}
Different from \cite{zhou2011throughput} and  \cite{xi2013} where only the approximation or bound on CDF of SIR at the eavesdropper was derived,  our result is derived in a simple   exact closed-form expression.
It is observed from \eqref{CDFeveAN_2} that the CDF of  ${\gamma _{\tt SIR}^{e,AN}}$ is an increasing function of $\lambda_s$ and $\lambda_p$, and a decreasing function of $\lambda_e$.

By substituting the CDF of ${\gamma^{s,AN}_{\tt SIR}}$ in \eqref{CDFsecan2} and the CDF of  ${\gamma^{e,AN}_{\tt SIR}}$ in \eqref{CDFeveAN_2} into \eqref{y_31}, we derive the average secrecy rate in the following theorem.
\textbf{\begin{theo}
With BF$\&$AN at the SU transmitters,  the average secrecy rate  is derived as
\begin{align}\label{y_32}
{{\bar R}_{se,AN}} = &\frac{1}{{\ln 2}}\int_0^\infty  {\frac{{\exp ( - \frac{{\pi {\lambda _e}}}{{{\Lambda _l}}}{x_2}^{ - \frac{2}{\alpha }}{{(\frac{{1 - \mu }}{{({N_s} - 1)\mu }}{x_2} + 1)}^{1 - {N_s}}})}}{{1 + {x_2}}}}
\nonumber\\
&\exp \big( { - {\Lambda _l}{x_2}^{\frac{2}{\alpha }}r_s^2} \big)\Big[ {1 + \sum\limits_{m = 1}^{{N_s} - 1} {\frac{{{{\left( {r_s^\alpha } \right)}^m}}}{{m!{{\big( { - 1} \big)}^m}}}} } \Big.\sum m!
\nonumber\\
&\Big. \prod\limits_{j = 1}^m {\frac{{{{(( - {\Lambda _l}{x_2}^{\frac{2}{\alpha }}){{({r_s})}^{2 - j\alpha }}\prod\limits_{k = 0}^{j - 1} {(\frac{2}{\alpha } - k)} )}^{{m_j}}}}}{{{m_j}!j{!^{{m_j}}}}}}  \Big]d{x_2},
\end{align}
where $\Lambda _l$ is given in \eqref{lambda}. Here, $P_s$ is the maximum permissive transmit power, which is  given in \eqref{Ps_AN}.
\end{theo}}



\begin{figure}[t!]
    \begin{center}
        \includegraphics[width=3.2in,height=2.6in]{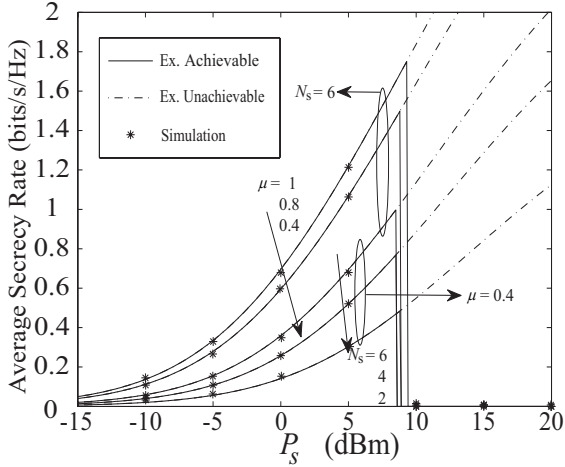}
        \caption{Average secrecy rate of a large scale spectrum sharing network with the transmit power adaptation
scheme. Parameters: $\lambda_e=\lambda_p=10^{-4}\;m^{-2}$, $\lambda_s=10^{-3}\; m^{-2}$, $\alpha=4$, $r_p=15 \;m$, $r_s=10 \; m$, ${P_p}=36$~dBm, and ${{\gamma _{th}^{\left\{ p \right\}} }}=0$~dBm.}
        \label{fig:1}
    \end{center}
\end{figure}

Note that the average secrecy rate given in \eqref{y_32} is applicable to  arbitrary $N_s$, $\mu$ and $\alpha$.

\subsubsection{Secrecy Outage Probability}
The secrecy outage is declared when the instantaneous  secrecy rate $R_{se}$ is less than the expected secrecy rate $R_s$. As such, the secrecy outage probability is defined as \cite{lifeng2014physical}
\begin{align}\label{SOP_1}
{P_{out}}\left( {{R_s}} \right) &= \Pr \left( {{R_{\tt se}} < {R_s}} \right)\nonumber \\
&\hspace{-1cm}=\int_0^\infty  {{f_{{\gamma _{\tt SIR}^{e,AN}}}}\left( {{x_2}} \right){F_{{\gamma _{\tt SIR}^{s,AN}}}}\left( {{2^{{R_s}}}\left( {1{\rm{ + }}{x_2}} \right) - 1} \right)} d{x_2}.
\end{align}

By substituting the probability density function (PDF) of $\gamma _{\tt SIR}^{e,AN}$ and CDF of $\gamma _{\tt SIR}^{s,AN}$ into \eqref{SOP_1}, we derive the secrecy outage probability  in the following theorem.
\textbf{\begin{theo}
With BF$\&$AN at the SU transmitters,  the secrecy outage probability is derived as
\begin{align}\label{SOP_2}
&{P_{out,AN}}\left( {{R_s}} \right) =   \int_0^\infty  {\frac{{\pi {\lambda _e}{x_2}^{ - \frac{2}{\alpha }}\big( {\frac{2}{\alpha }{x_2}^{ - 1}\big( {\frac{{1 - \mu }}{{\left( {{N_s} - 1} \right)\mu }}{x_2} + 1} \big) + 1} \big)}}{{{\Lambda _l}{{\big( {\frac{{1 - \mu }}{{\left( {{N_s} - 1} \right)\mu }}{x_2} + 1} \big)}^{{N_s}}}}}}
\nonumber\\
\hspace{-0.1cm}& \exp \big( { - \frac{{\pi {\lambda _e}}}{{{\Lambda _l}}}{x_2}^{ - \frac{2}{\alpha }}{{\big( {\frac{{1 - \mu }}{{\left( {{N_s} - 1} \right)\mu }}{x_2} + 1} \big)}^{1 - {N_s}}}} \big)\Big[ {1 - } \Big.
\nonumber\\
&\exp \big( { - {\Lambda _3}{{\left( {{2^{{R_s}}}\big( {1{\rm{ + }}{x_2}} \big) - 1} \right)}^{\frac{2}{\alpha }}}r_s^2} \big)\Big( {1 + \sum\limits_{m = 1}^{{N_s} - 1} {\frac{{{{\big( {r_s^\alpha } \big)}^m}}}{{m!{{\big( { - 1} \big)}^m}}}} \sum m!} \Big.
\nonumber\\
 &\Big. {\Big. \prod\limits_{j = 1}^m {\frac{{{{(( - {\Lambda _l}{{({2^{{R_s}}}(1{\rm{ + }}{x_2}) - 1)}^{\frac{2}{\alpha }}}){{({r_s})}^{2 - j\alpha }}\prod\limits_{k = 0}^{j - 1} {(\frac{2}{\alpha } - k)} )}^{{m_j}}}}}{{{m_j}!j{!^{{m_j}}}}}}  \Big)} \Big]d{x_2},
\end{align}
where $\Lambda _l$ is given in \eqref{lambda}. Here, $P_s$ is the maximum permissive transmit power, which is  given in \eqref{Ps_AN}.
\end{theo}}

\subsection{Numerical Examples for BF $\&$ AN}

\begin{figure}[t!]
    \begin{center}
        \includegraphics[width=3.2in,height=2.6in]{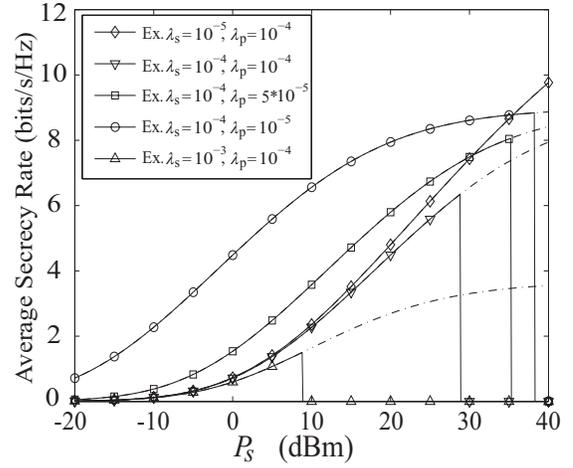}
        \caption{Average secrecy rate of a large scale spectrum sharing network with the transmit power adaptation
scheme. Parameters: $\lambda_e=10^{-4}$ $m^{-2}$, $N_s=6$, $\alpha=4$, $r_p=15$ $m$, $r_s=10$ $m$, $\mu=0.8$, ${P_p}=36$~dBm, and ${{\gamma _{th}^{\left\{ p \right\}} }}=0$~dBm.}
        \label{fig:2}
    \end{center}
\end{figure}

 \begin{enumerate}
 \item {\emph{Average Secrecy Rate Boundary}}
\smallskip
 \\ Fig.~\ref{fig:1} and  Fig.~\ref{fig:2} plot the average secrecy rate of large scale underlay spectrum sharing networks under the  primary network's QoS constraint $\rho^{\left\{ p \right\}}_{ out} =
0.15$ with the transmit power adaptation
scheme.
From these  figures, we see that
the exact analytical curves are well validated by  Monte Carlo simulations. The solid lines represent the operational achievable average secrecy rate where the  primary network's QoS constraint is always satisfied, i.e., $P_{out}^{pri,AN}\big( {\gamma _{th}^{\left\{ p \right\}}} \big) \le 0.15 $. The dashed lines represent the unachievable average secrecy rate where the   primary network's QoS constraint is violated, i.e., $P_{out}^{pri,AN}\big( {\gamma _{th}^{\left\{ p \right\}}} \big)  >  0.15 $.
{ We named the solid line as ``average secrecy rate boundary".}

\smallskip
\item \emph{Impact of $N_s$ and $\mu$ on the average secrecy rate}
\smallskip
\\
Fig.~\ref{fig:1} plots the average secrecy rate versus the SU's transmit power with various number of transmit antennas $N_s$ at the SU  and power allocation factor $\mu$, and we consider the same density for PUs, SUs, and eavesdroppers. The exact analytical curves are obtained from  \eqref{y_32}.
 We find that for  fixed  $\mu=0.4$, the average secrecy rate increases with increasing $N_s$.

\smallskip
\item \emph{Impact of $\lambda_s$ and $\lambda_p$ on the average secrecy rate}
\smallskip
\\
Fig.~\ref{fig:2} plots the average secrecy rate versus $P_s$ for various densities of PUs and SUs. We observe that  there is a shift of the ``average secrecy rate  wall"  to the left with increasing the density of PUs and SUs. This can be predicted from  \eqref{Ps_AN} that ${P_s^{\max }}$ is a decreasing function of $\lambda_p$ and $\lambda_s$. As expected,  the average secrecy rate decreases with increasing the density of SUs and PUs, due to the increased aggregate interference from the SUs and the PUs.

\smallskip
\item \emph{Optimal $\mu$  for the average secrecy rate}
\smallskip
\\

\begin{figure}[t!]
    \begin{center}
        \includegraphics[width=3.2in,height=2.6in]{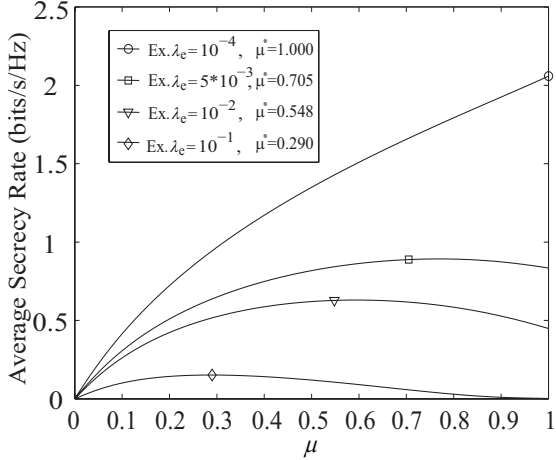}
        \caption{ Average secrecy rate of a large scale spectrum sharing network. Parameters: $\lambda_p=10^{-4}\;m^{-2}$, $ \lambda_s=10^{-3}\;m^{-2}$, $\alpha=3$, $r_p=15\;m$, $r_s=10\;m$, $N_s=6$, ${P_p}=15$~dBm, and ${{\gamma _{th}^{\left\{ p \right\}} }}=0$~dBm.}
        \label{fig:3}
    \end{center}
\end{figure}

Fig.~\ref{fig:3}  plots the average secrecy rate versus  $\mu$ for various densities of eavesdroppers $\lambda_e$. Here, we  use the maximum permissive transmit power to transmit the signal from SU, which is  given by \eqref{y_32}, and we set $P_s=P_s^{max}$ and  $\rho^{\left\{ p \right\}}_{ out}=0.1$. The triangles represent the maximum achievable average secrecy rate. For the scenarios where the density of eavesdroppers is higher than the density of  SUs, the average secrecy rate first increases  and then decreases with increasing $\mu$.  An optimal power allocation factor  $\mu^*$  exists at which the maximum average secrecy rate is achieved. For the region $\mu<\mu^*$, we  see that increasing the power  allocated to the useful signal ensures more message delivery  (increasing $C_{\tt su}$)  and plays a dominant role in improving the average secrecy rate; for the region $\mu>\mu^*$,  reducing the power allocated to the AN increases $C_{\tt E}$, and thus degrades the average secrecy rate. We conclude that a  tradeoff exists between increasing the capacity of secondary channel and decreasing the capacity of eavesdropping channel.
 Interestingly, we see from  Fig.~\ref{fig:3} that $\mu^*$ varies for different $\lambda_e$. We find that less power should be allocated to the AN for a network with less dense eavesdroppers. Out of expectation, the equal power allocation may not be a good strategy to achieve the maximum average secrecy rate.


\begin{figure}[t!]
    \begin{center}
        \includegraphics[width=3.2in,height=2.6in]{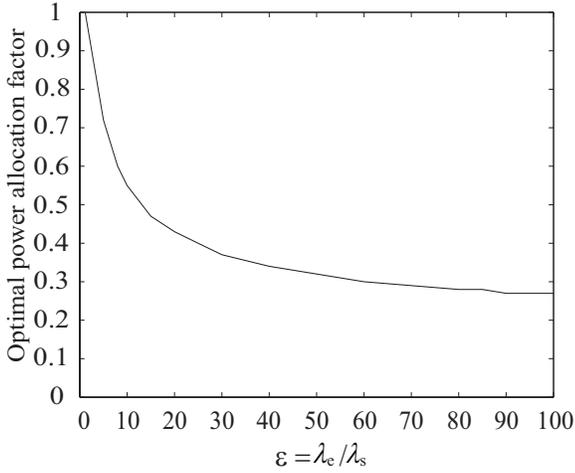}
        \caption{ Optimum $\mu$ for maximum average secrecy rate versus $\varepsilon$. Parameters:
         $\lambda_s=10^{-3}\;m^{-2}$, $\lambda_p=\lambda_s/10$, $N_s = 4$,  $\rho^{\left\{ p \right\}}_{ out}=0.1$, $\alpha=3$, $r_p=15\;m,$ $r_s=10\;m,$ ${P_p}=15$~dBm, and ${{\gamma _{th}^{\left\{ p \right\}} }}=0$~dBm.}
        \label{fig:4}
    \end{center}
\end{figure}


\smallskip
\item \emph{Impact of density ratio on the optimal $\mu$}
\smallskip
\\
To better illustrate the relationship between the optimal power allocation factor and the density of SUs and eavesdroppers.  We first define the ratio between $\lambda_e$ and $\lambda_s$ as  $\varepsilon  = {{{\lambda _e}} \mathord{\left/
 {\vphantom {{{\lambda _e}} {{\lambda _s}}}} \right.
 \kern-\nulldelimiterspace} {{\lambda _s}}}$. In Fig.~\ref{fig:4}, we  plot  $\mu^*$ versus the density ratio $\varepsilon$.  We set   $P_s=P_s^{max}$, $\lambda_s=10^{-3}\; m^{-2}$, and $\lambda_p=10^{-4}\; m^{-2}$.  We find that 1) The power allocated to AN should be increased with increasing the density ratio between the eavesdroppers and the SUs $\varepsilon$ to achieve the optimal average secrecy rate;
 2) For extremely low density ratio $\varepsilon$,  all of the power should be allocated to information signal without injecting AN to achieve the maximum average secrecy rate. This reveals that improving the information delivery is more important than combating the eavesdropping in this scenario.

\smallskip
\item \emph{Impact of $N_s$ and $\alpha$ on the secrecy outage probability}
\smallskip
\\
  Fig.~\ref{fig:5} plots the secrecy outage probability versus $\mu$ for various number of antennas at SU transmitter $N_s$.
  The exact analytical curves are obtained from  \eqref{SOP_2}, which are well validated by  Monte Carlo simulations.
   We assume $P_s=P_s^{max}$. In this setting, we see that the secrecy outage probability decreases with increasing $\mu$,
   and when $\mu$ approaches 1, the lowest secrecy outage probability is achieved. This is because when the density
   of eavesdroppers is small compared to that of SU,  the effect of delivering information overtakes the effect of
   combating the eavesdropping. As expected, the secrecy outage probability decreases with
   increasing $N_s$, which is due to the array gains brought by additional antennas.


\begin{figure}[t!]
    \begin{center}
    \vspace{-0.75cm}
        \includegraphics[width=3.4in,height=2.9in]{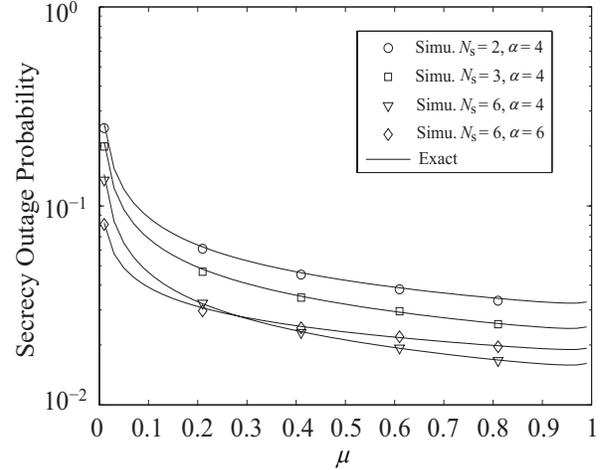}
        \caption{Secrecy outage probability versus $\mu$ for various  $N_s$ and $\alpha$. Parameters: $\rho^{\left\{ p \right\}}_{ out}=0.1$, $\lambda_e=10^{-4}\; m^{-2}$, $\lambda_p=10^{-4}\; m^{-2}$, $\lambda_s=10^{-3}\; m^{-2}$, $r_p=6\;m,$ $r_s=3\;m,$ $N_S=6,$ $Rs=1$, ${P_p}=15$~dBm, and ${{\gamma _{th}^{\left\{ p \right\}} }}=0$~dBm.}
        \label{fig:5}
    \end{center}
\end{figure}

 \end{enumerate}

\begin{figure}[t!]
\centering
\subfigure[Parameters: $\lambda_s=\lambda_e=10^{-4}\; m^{-2}$, $\lambda_p=10^{-5}\; m^{-2}$, $\rho^{\left\{ p \right\}}_{ out}=0.1$, $\mu=0.4$, $r_p=15\;m$, $r_s=10\;m$, and $\alpha=4$]{
\label{Fig.sub.61}
\includegraphics[width=3 in]{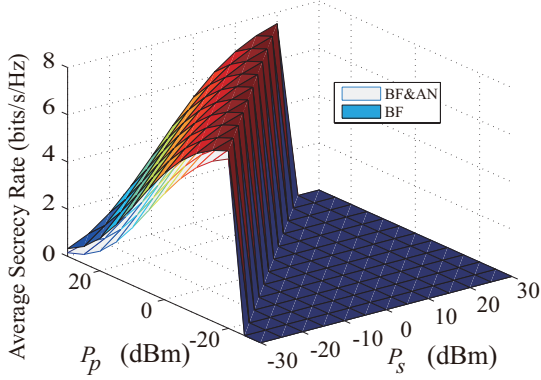}}
\subfigure[ Parameters: $\lambda_p=\lambda_s=10^{-6}\; m^{-2}$, $\lambda_e=10^{-5}\; m^{-2}$, $\rho^{\left\{ p \right\}}_{ out}=0.1$, $\mu=0.4$, $r_p=15\;m$, $r_s=10\;m$, and $\alpha=3$]{
\label{Fig.sub.62}
\includegraphics[width=3 in]{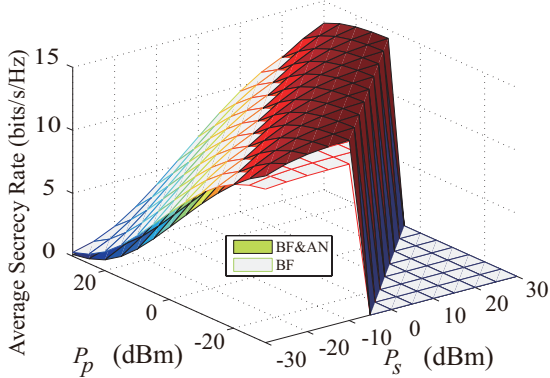}}
\caption{Comparison of average secrecy rate versus $P_s$ and $P_p$ between BF$\&$AN and BF.}
\label{Fig2}
\end{figure}

\subsection{Numerical examples for the comparison between BF$\&$AN and BF}

In this subsection, we compare the secrecy performance of our proposed network with BF$\&$AN to that with  BF, and examine the potential benefits of AN on the secrecy performance.
Note that BF can be viewed as a special case of BF$\&$AN with $\mu=1$.

In Fig.~\ref{Fig.sub.61} and Fig.~\ref{Fig.sub.62}, we plot the operational  achievable average secrecy rate region for the large scale spectrum sharing network with BF$\&$AN and   BF. We see the same  permissive transmit power region for   BF$\&$AN  and   BF in each figure. This is because, from the typical PU receiver's perspective,
both  AN and the information signal transmitted from SU  are viewed as  interference, which is equivalent to the case of BF. In both figures, we notice that  the same unachievable average secrecy rate region located in $
{P_s} \in \left( {0,30} \right)$~dBm  with  $ {P_p} \in \left( { 0, - 30} \right)$~dBm. 	This can be explained by the fact that the QoS constraint is severely violated in this setting when the aggregate interference is much higher compared to the  useful signal received at PU.

In contrast to the fact that it is often beneficial to emit AN on top of the information-bearing signal in the physical layer
security model with fixed nodes \cite{hong2013enhancing},
 we see from Fig.~\ref{Fig.sub.61} that  BF  outperforms  BF$\&$AN  or has the same performance as BF$\&$AN in all operational region. This is because the strong aggregate interference from the PUs
 overtakes the effect of the AN
generated by SU. In this case, more power needs to be allocated to transmit information signal at SUs to contend with  the interference from PUs.   Interestingly, Fig.~\ref{Fig.sub.62} shows that   BF$\&$AN  outperforms  BF  in some regions, owing to the fact that the effect of  AN generated by SU  overtakes the relatively low aggregate interference from PUs. In this case, more power should be allocated to transmit AN to disrupt the eavesdropping.

%
%

\section{Large Antenna arrays analysis}
\label{asymptotic secrecy}

In this section, we study the asymptotic secrecy performance of the large scale spectrum sharing networks where the SU  transmitters are equipped with large antenna arrays. We examine  the asymptotic behavior of the average secrecy rate and the secrecy outage probability, when the  number of antennas at the SU transmitters goes to infinity.


We first present the {Lemma 4} based on the law of large numbers as follows:
\begin{lemma}
 $\mathop {\lim }\limits_{{N_s} \to \infty } {\left\| {{{\bf{h}}_z}} \right\|^2} = {N_s}$, and $\mathop {\lim }\limits_{{N_s} \to \infty } {\left\| {{{\bf{h}}_{i,z}}{{\bf{G}}_{i,{s_i}}}} \right\|^2} = {N_s} - 1$.
\end{lemma}
\begin{proof}
This is due to the fact that ${\left\| {{{\bf{h}}_z}} \right\|^2}\mathop  \sim Gamma\left( {{N_s},1} \right)$ and   ${\left\| {{{\bf{h}}_{i,z}}{{\bf{G}}_{i,{s_i}}}} \right\|^2} \mathop  \sim Gamma\left( {{N_s} - 1,1} \right)$.
\end{proof}

By using {Lemma 4},  we  rewrite the SIR at the typical PU in \eqref{PU_SIR_AN}   as
\begin{align}\label{SINRasypri}
\gamma _{\tt SIR}^{p,\infty } \mathop  \sim \limits^d  \frac{{{{\left| {{h_{{p_0}}}} \right|}^2}{r_p}^{ - \alpha }}}{{{I_{p,{p_0}}} + \eta {I_{s,{p_0},\infty }}}},
\end{align}
where 
{
\begin{align}
{I_{p,{p_0}}} = {\sum _{j \in {\Phi _p}\backslash \left\{ 0 \right\}}}{\left| {{h_{j,{p_0}}}} \right|^2}{\left| {{X_{j,{p_0}}}} \right|^{ - \alpha }}
\end{align}
  and 
 \begin{align} 
  {I_{s,{p_0},\infty }} = {\sum _{i \in {\Phi _s}}}\Big[ {\mu {{\big| {{{\bf{h}}_{i,{p_0}}}\frac{{{\bf{h}}_{i,{s_i}}^\dag }}{{\big\| {{\bf{h}}_{i,{s_i}}^\dag } \big\|}}} \big|}^2} + \big( {1 - \mu } \big)} \Big]{\big| {{X_{i,{p_0}}}} \big|^{ - \alpha }}.
  \end{align}}

For large $N_s$, the SIR at typical SU  is given as
\begin{align}\label{SINRasysec}
\gamma _{\tt SIR}^{s,\infty} \mathop  \sim \limits^d  \frac{{\mu {N_s}{r_s}^{ - \alpha }}}{{{I_{s,{s_0},\infty }} + {\eta ^{ - 1}}{I_{p,{s_0}}}}},
\end{align}
where 
{
\begin{align}
{I_{p,{s_0}}} = {\sum _{j \in {\Phi _p}}}{\left| {{h_{j,{s_0}}}} \right|^2}{\left| {{X_{j,{s_0}}}} \right|^{ - \alpha }}
\end{align}
 and  
 \begin{align} 
 {I_{s,{s_0},\infty }} = {\sum _{i \in {\Phi _s}\backslash \left\{ 0 \right\}}}\Big[ {\mu {{\big| {{{\bf{h}}_{i,{s_0}}}\frac{{{\bf{h}}_{i,{s_i}}^\dag }}{{\left\| {{\bf{h}}_{i,{s_i}}^\dag } \right\|}}} \big|}^2} + \left( {1 - \mu } \right)} \Big]{\left| {{X_{i,{s_i}}}} \right|^{ - \alpha }}.
 \end{align}} 
 From \eqref{SINRasysec}, we find that the received SIR at typical SU scale by $N_s$.

For large $N_s$, the SIR  at the most detrimental eavesdropper   is given as
\begin{align}\label{Eve_SIR_AN_inf}
&\gamma _{\tt SIR}^{e,\infty} =  \mathop {\max }\limits_{{{e_k} \in {\Phi _e}}} \left\{ {\gamma _{\tt SIR}^{{e_k},\infty}} \right\},
\end{align}
where
\begin{align}\label{SINRasyeav_inf}
\gamma _{\tt SIR}^{{e_k},\infty }  \mathop  \sim \limits^d  \frac{{\mu {{\big| {{{\bf{h}}_{{0},{e_k}}}\frac{{{\bf{h}}_{i,{s_i}}^\dag }}{{\left\| {{\bf{h}}_{i,{s_i}}^\dag } \right\|}}} \big|}^2}{{\left| {{X_{e_k}}} \right|}^{ - \alpha }}}}{{{I_{{s_0},{e_k},\infty }} + {\eta ^{ - 1}}{I_{p,{e_k}}} + \left( {1 - \mu } \right){{\left| {{X_{e_k}}} \right|}^{ - \alpha }}}},
\end{align}
where 
{\begin{align}
{I_{p,{e_k}}} = {\sum _{j \in {\Phi _p}}}{\left| {{h_{j,{e_k}}}} \right|^2}{\left| {{X_{j,{e_k}}}} \right|^{ - \alpha }}
\end{align}
 and 
 \begin{align}
 {I_{{s_0},{e_k},\infty }} = {\sum _{i \in {\Phi _s}\backslash \left\{ 0 \right\}}}\Big[ {\mu {{\left| {{{\bf{h}}_{i,{e_k}}}\frac{{{\bf{h}}_{i,{s_i}}^\dag }}{{\left\| {{\bf{h}}_{i,{s_i}}^\dag } \right\|}}} \right|}^2} + \left( {1 - \mu } \right)} \Big]{\left| {{X_{i,{e_k}}}} \right|^{ - \alpha }}.
 \end{align}}

Based on the SIR at the typical PU in \eqref{SINRasypri},  with the help of the Laplace transform in \cite[eq. (8)]{haenggi2009stochastic}, and similar method provided in the proof for the {\bf{Theorem 1}}, we  present the permissive transmit power region at the SU transmitter at large $N_s$  in the following proposition.

\begin{ppro}
With BF$\&$AN at the SU transmitters,  the permissive transmit power region for the SU transmitter at large $N_s$ is given as ${P_s} \in \left( {0,} \right.\left. {P_s^{\max }} \right]$, where
\begin{align}\label{Psinf_AN_asy}
&{P_s^{max}}  ={\Big[ {-\Theta {{\big( {\int_0^\infty  {{{\big( {\mu t + \big( {1 - \mu } \big)} \big)}^{\frac{2}{\alpha }}}{e^{ - t}}dt} } \big)}^{ - 1}}{\lambda _s}^{ - 1}} \Big]^{\frac{\alpha }{2}}}{P_p},
\end{align}
and $\Theta$ is given by \eqref{theta}.
%
\end{ppro}

\bigskip
To facilitate the analysis of the average secrecy rate and the secrecy outage probability, we need to  first derive the asymptotic CDFs of $\gamma _{\tt SIR}^{e,\infty}$ and $\gamma _{\tt SIR}^{s,\infty}$.
Using the method presented in Appendix B, we derive the asymptotic CDF of  $\gamma _{\tt SIR}^{e,\infty}$ given in \eqref{Eve_SIR_AN_inf} as
\begin{align}\label{CDFeveAN_2large}
&\hspace{-0.3cm}F_{\gamma _{\tt SIR}^{e,\infty}} \big( {\gamma _{th}^{\{ e\} }} \big) = \exp \Big( { - \frac{{\pi {\lambda _e}{e^{\big( {1 - {\mu ^{ - 1}}} \big)\gamma _{th}^{\{ e\} }}}}}{{\Xi \Gamma \big( {1 - \frac{2}{\alpha }} \big)}}{{\big( {\frac{{\mu {P_s} }}{{\gamma _{th}^{\{ e\} }}{P_p}}} \big)}^{\frac{2}{\alpha }}}} \Big),
\end{align}
where
\begin{align}\label{Xi_1}
\Xi  = {\lambda _p}\Gamma \big( {1 + \frac{2}{\alpha }} \big) + {\lambda _s}{\big({\frac{{{P_s}}}{{{P_p}}}}\big)^{\frac{2}{\alpha }}}\int_0^\infty  {{{\big( {\mu t + \big( {1 - \mu } \big)} \big)}^{\frac{2}{\alpha }}}{e^{ - t}}dt} .
\end{align}

To derive the asymptotic CDF of  $\gamma _{\tt SIR}^{s,\infty}$, we first present
\begin{align}\label{CDFsecanprf1large}
&F_{\gamma _{\tt SIR}^{s,\infty}} \left( {\gamma _{th}^s} \right) =  \int_{\frac{{\mu {N_s}}}{{\gamma _{th}^s{r_{{s_0}}}^\alpha }}}^\infty  {{f_{{I_{\sec ,\infty }}}}\left( x \right)dx}  ,
\end{align}
where
\begin{align}\label{CDFsecanprf2large}
{I_{\sec ,\infty }} =& \mu {\sum _{i \in {\Phi _s}\backslash \big\{ 0 \big\}}}{\big| {{{\bf{h}}_{i,{s_0}}}\frac{{{\bf{h}}_{i,{s_i}}^\dag }}{{\big\| {{\bf{h}}_{i,{s_i}}^\dag } \big\|}}} \big|^2}{\big| {{X_{i,{s_i}}}} \big|^{ - \alpha }} + \big( {1 - \mu } \big)
\nonumber\\
&{\sum _{i \in {\Phi _s}\backslash \left\{ 0 \right\}}}{\left| {{X_{i,{s_i}}}} \right|^{ - \alpha }} + {\big({\frac{{{P_s}}}{{{P_p}}}}\big)^{ - 1}}{\sum _{j \in {\Phi _p}}}{\left| {{h_{j,{s_0}}}} \right|^2}{\left| {{X_{j,{s_0}}}} \right|^{ - \alpha }}.
\end{align}

In \eqref{CDFsecanprf1large}, ${f_{I_{\sec ,\infty }} }\left( x \right)$ is the inverse Laplace transform of ${\mathcal{L}_{I_{\sec ,\infty }}}\left( s \right)$, which can be expressed as  ${f_{I_{\sec ,\infty }} }\left( x \right)={\mathcal{L}_{I_{\sec ,\infty }}^{-1}}\left( s \right)$.
Due to the intractability of this inverse Laplace transform, some alternative ways have been proposed, such as using numerical inversion to evaluate ${\mathcal{L}_{I_{\sec ,\infty }}^{-1}}\left( s \right)$ \cite{Hollenbeck1998}, or  the log-normal approximations to approximate ${f_{I_{\sec ,\infty }} }\left( x \right)$.
However, in our case, there exists singularity at $\left| {{X_{i,{s_i}}}} \right|=\left| {{X_{j,{s_0}}}} \right|=0$, thus  the mean and variance of ${I_{\sec ,\infty }}$ derived from the moment generating function diverge \cite{Geraci2014,venkataraman2006shot}, which renders the derivation of the PDF of ${I_{\sec ,\infty }}$. Alternatively, we utilize the Gil-Pelaez theorem~\cite{wendel1961} to facilitate the derivation of the  asymptotic CDF of SIR at the typical SU in the following lemma.
\begin{lemma}
With BF$\&$AN at the SU transmitters,  the asymptotic CDF of SIR at the typical SU at large ${N_s}$ is given as
\begin{align}\label{CPT_gpt_OPlarge}
{F_{\gamma _{\tt SIR}^{s,\infty }}}\big( {\gamma _{th}^s} \big)& = 1 - {F_{{I_{\sec ,\infty }}}}\big( {\frac{{\mu {N_s}}}{{\gamma _{th}^s{r_{{s_0}}}^\alpha }}} \big)
\nonumber\\
&= \frac{1}{2} + \frac{1}{\pi }\int_0^\infty  {\frac{{{\rm{Im}}\big[ {{e^{ - \frac{{jw\mu {N_s}}}{{\gamma _{th}^s{r_{{s_0}}}^\alpha }}}}{\varphi ^*}\left( w \right)} \big]}}{w}} dw,
\end{align}
where ${F_{{I_{\sec ,\infty }}}}\left( x\right)$ is the CDF of ${{I_{\sec ,\infty }}}$, $j=\sqrt{\left(-1\right)}$, and $\varphi\left(w\right)$ is the conjugate of the characteristic function, which is given by
\begin{align}\label{CPT_gpt_OP_2large}
\varphi\big(w\big)=&{\exp \big( { - \pi \Xi \Gamma \big( {1 - \frac{2}{\alpha }} \big){\eta ^{ - \frac{2}{\alpha }}}{{\big( {jw} \big)}^{\frac{2}{\alpha }}}} \big)}.
\end{align}
\end{lemma}

Since we can not derive the closed form expression for the general form for the PDF of ${I_{\sec ,\infty }}$, we present the special case for the path loss component $\alpha=4$. In the following corollary, we derive the asymptotic CDF of SIR for the typical secondary user with $\alpha=4$.
\begin{cor}
With BF$\&$AN at the SU transmitters  and $\alpha=4$,  the asymptotic CDF of SIR at the typical SU is derived as
\begin{align}\label{CDFsecan2large}
&F_{\gamma _{\tt SIR}^{s,AN}}^\infty \big( {\gamma _{th}^s} \big) = \Phi \big( {\frac{{\pi {\Xi }}}{2}\sqrt {\frac{{\pi \gamma _{th}^s{r_{{s_0}}}^\alpha }}{{\mu {N_s}}}} } \big),
\end{align}
where
\begin{align}\label{phi}
\Phi \left( x \right) = \frac{1}{{\sqrt \pi  }}\int_0^{{x^2}} {\frac{{{e^{ - t}}}}{{\sqrt t }}} dt.
\end{align}
\end{cor}
Note that our derived asymptotic CDF of SIR at the typical SU for $\alpha=4$ is in exact closed-form.

\subsection{Average Secrecy Rate}

Based on the CDF of SIR  at the most detrimental eavesdroppers in \eqref{CDFeveAN_2large} and the CDF of SIR at the typical SU in \eqref{CPT_gpt_OPlarge},  we derive the general case of the asymptotic average secrecy rate using \eqref{y_31} in the following proposition.

\begin{ppro}
With BF$\&$AN at the SU transmitters, the asymptotic average secrecy rate at large ${N_s} $ is derived as
\begin{align}\label{secrecyrate_AN_asy_gene}
&\bar C_{se}^\infty = \nonumber\\& \frac{1}{{\ln 2}}\int_0^\infty  {\frac{1}{{1 + {x_2}}}\exp \big( { - \frac{{\pi {\lambda _e}{e^{\big( {1 - {\mu ^{ - 1}}} \big){x_2}}}}}{{\Xi \Gamma \big( {1 - \frac{2}{\alpha }} \big)}}{{\big( {\frac{{\mu P_s }}{{{x_2}{P_p}}}} \big)}^{\frac{2}{\alpha }}}} \big)} \Big[ {\frac{1}{2} - \frac{1}{\pi }} \Big.
\nonumber\\&
\Big.\int_0^\infty  {{\rm{Im}}[\frac{{{{(\exp ( - \pi \Xi \Gamma (1 - \frac{2}{\alpha }){{(\frac{{{P_s}}}{{{P_p}}})}^{ - \frac{2}{\alpha }}}{{(jw)}^{\frac{2}{\alpha }}}))^*}}}}{{{e^{\frac{{jw\mu {N_s}}}{{{x_2}{r_{{s_0}}}^\alpha }}}}}}]} {\rm{ }}\frac{1}{w}dw\Big]d{x_2}.
\end{align}
\end{ppro}

Having  \eqref{CDFeveAN_2large} and \eqref{CDFsecan2large}, we derive the asymptotic average secrecy rate for the special case of  $\alpha=4$ in the following corollary.

\begin{cor}
With BF$\&$AN at the SU transmitters and $\alpha=4$, the asymptotic average secrecy rate at large ${N_s} $ is derived as
\begin{align}\label{secrecyrate_AN_asy}
\bar C_{se}^\infty  = & \frac{1}{{\ln 2}}\int_0^\infty  {\frac{1}{{1 + {x_2}}}} \exp \big( { - \frac{{\pi {\lambda _e}{e^{\big( {1 - {\mu ^{ - 1}}} \big){x_2}}}}}{{\Xi \Gamma \big( {1 - \frac{2}{\alpha }} \big)}}{{\big( {\frac{\mu }{{{x_2}}}} \big)}^{\frac{2}{\alpha }}}} \big)
\nonumber\\& \big( {1 - \Phi \big( {\frac{{\pi \Xi }}{2}\sqrt {\frac{{\pi {x_2}}}{{\mu {N_s}{r_{{s_0}}}^{ - \alpha }}}} } \big)} \big)d{x_2}.
\end{align}
\end{cor}

\subsection{Secrecy Outage Probability}

We then turn our attention to the asymptotic secrecy outage probability. We take the derivative of the asymptotic CDF of SIR at the most detrimental eavesdroppers in \eqref{CDFeveAN_2large}, and substitute it with  the asymptotic CDF of SIR at the typical SU in \eqref{CDFsecan2large} into \eqref{SOP_1}, to yield the general case of the asymptotic secrecy outage probability in the  following proposition.

\begin{ppro}
With BF$\&$AN at the SU transmitters, the asymptotic  secrecy outage probability  at large ${N_s} $ is derived as
\begin{align}\label{secrecyop_AN_asy_gene}
&P_{out,AN}^\infty \big( {{R_s}} \big) = \nonumber\\
& \hspace{0.3cm}\int_0^\infty  {\frac{{\pi {\lambda _e}{{(\mu {{{P_s}} \mathord{\left/
 {\vphantom {{{P_s}} {{P_p}}}} \right.
 \kern-\nulldelimiterspace} {{P_p}}})}^{\frac{2}{\alpha }}}}}{{\Xi \Gamma (1 - \frac{2}{\alpha }){x_2}^{\frac{2}{\alpha }}}}\exp ( - \frac{{\pi {\lambda _e}{e^{(1 - {\mu ^{ - 1}}){x_2}}}}}{{\Xi \Gamma (1 - \frac{2}{\alpha })}}{{(\frac{{\mu {P_s}}}{{{x_2}{P_p}}})}^{\frac{2}{\alpha }}})}
\nonumber\\& \hspace{0.3cm}
{e^{\big( {1 - {\mu ^{ - 1}}} \big){x_2}}}\Big( {\big( {1 - {\mu ^{ - 1}}} \big) - \big( { - \frac{2}{\alpha }} \big){x_2}^{ - 1}} \Big)\Big[ {\frac{1}{2} + \frac{1}{\pi }} \Big.
\nonumber\\& \hspace{0.3cm} \int_0^\infty  {{\rm{Im}}[\frac{{{{(\exp ( - \pi \Xi \Gamma (1 - \frac{2}{\alpha }){{\left( {{{{P_s}} \mathord{\left/
 {\vphantom {{{P_s}} {{P_p}}}} \right.
 \kern-\nulldelimiterspace} {{P_p}}}} \right)}^{ - \frac{2}{\alpha }}}{{(jw)}^{\frac{2}{\alpha }}}))^*}}}}{{{e^{\frac{{jw\mu {N_s}}}{{({2^{{R_s}}}(1{\rm{ + }}{x_2}) - 1){r_{{s_0}}}^\alpha }}}}}}]}
\nonumber\\& \hspace{0.3cm}
\left. {\frac{1}{w}dw} \right]d{x_2}.
\end{align}
\end{ppro}

\begin{figure}[t!]
    \begin{center}
        \includegraphics[width=3.2in,height=2.6in]{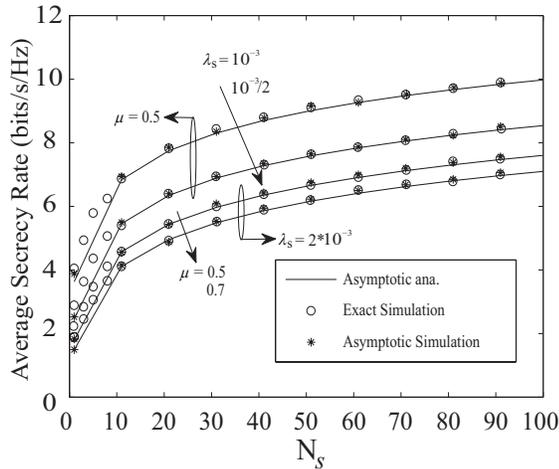}
        \caption{Asymptotic average secrecy rate versus $N_s$. Parameters: $\rho^{\left\{ p \right\}}_{ out}=0.1$,  $\lambda_p = 10^{-4}\; m^{-2}$,  $\lambda_e = 10^{-4}\; m^{-2}$, $r_p=6\;m$, $\mu=0.7$, $\alpha =3$, $r_s=3\;m$, ${P_p}=15$~dBm, and ${{\gamma _{th}^{\left\{ p \right\}} }}=6$~dBm. }
        \label{fig:7}
    \end{center}
\end{figure}

\bigskip
Based on \eqref{CDFsecan2large}, we derive the secrecy outage probability  for $\alpha=4$ as a special case in the following corollary.
\begin{cor}
With BF$\&$AN at the SU transmitters and $\alpha=4$, the asymptotic  secrecy outage probability  at large ${N_s} $ is derived as
\begin{align}\label{y_33}
&P_{out,AN}^\infty \big( {{R_s}} \big) = \nonumber\\
& \hspace{0.5cm} \int_0^\infty  {\frac{{\pi {\lambda _e}{\mu ^{\frac{2}{\alpha }}}}}{{\Xi \Gamma \big( {1 - \frac{2}{\alpha }} \big)}}\exp \big( { - \frac{{\pi {\lambda _e}{e^{\big( {1 - {\mu ^{ - 1}}} \big){x_2}}}}}{{\Xi \Gamma \big( {1 - \frac{2}{\alpha }} \big)}}{{\big( {\frac{\mu }{{{x_2}}}} \big)}^{\frac{2}{\alpha }}}} \big)}
\nonumber\\
& \hspace{0.5cm}{x_2}^{ - \frac{2}{\alpha }}{e^{\big( {1 - {\mu ^{ - 1}}} \big){x_2}}}\Big( {\big( {1 - {\mu ^{ - 1}}} \big) - \big( { - \frac{2}{\alpha }} \big){x_2}^{ - 1}} \Big)
\nonumber\\
& \hspace{0.5cm}\Phi \big( {\frac{{\pi \Xi }}{2}\sqrt {\frac{{\pi \big( {{2^{{R_s}}}\big( {1{\rm{ + }}{x_2}} \big) - 1} \big)}}{{\mu {N_s}{r_{{s_0}}}^{ - \alpha }}}} } \big)d{x_2}.
\end{align}
\end{cor}

\bigskip
\subsection{Numerical examples for the asymptotic secrecy performance of  BF$\&$AN}
\smallskip
Fig.~\ref{fig:7} plots the asymptotic average secrecy rate of large scale spectrum sharing networks with BF$\&$AN for various power allocation factor $\mu$ and $\lambda_s$. We assume $P_s= P_s^{max}$. The analytical results of asymptotic secrecy rate plotted using \eqref{secrecyrate_AN_asy_gene}  are in precise agreement with  the simulation points of asymptotic secrecy rate.
It is also shown that  the asymptotic average secrecy rate converges to the exact average secrecy rate at large $N_s$.

We observe that the average secrecy rate  increases with increasing $N_s$. This can be indicated by \eqref{SINRasysec} that the received SIR at the typical SU proportionally increases  with $\mu N_s$. For the same $\mu$, to achieve the same average secrecy rate, there exists  antenna gaps between the curves with different density of SU. This antenna gap quantifies how  many additional antennas needed to be employed  at the SU transmitter  to achieve the same average secrecy rate when the network  double its density of SU.

Fig.~\ref{fig:8} plots the asymptotic secrecy outage probability versus  $N_s$. The analytical results of asymptotic outage probability plotted using \eqref{secrecyop_AN_asy_gene}  are in precise agreement with  the simulation points of asymptotic outage probability. Furthermore, the asymptotic secrecy outage probability converges to the exact  secrecy outage probability at large $N_s$.
We see that the secrecy outage probability decreases  with increasing $N_s$, due to the increase of the array gains at the SU receiver. We also see that the secrecy outage probability decreases with increasing $\mu$, which reflects that for the scenario with relatively less dense eavesdroppers,
  more power should be allocated to transmit useful information to the SU receiver for the information delivery enhancement.

  \begin{figure}[t!]
\vspace{-0.18in}
    \begin{center}
        \includegraphics[width=3.2in,height=2.6in]{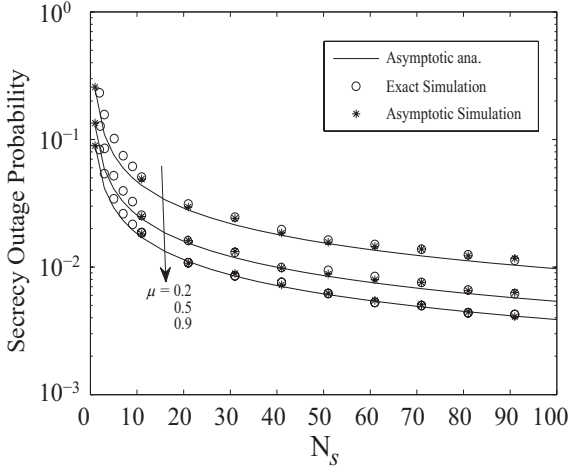}
        \vspace{0.01in}\caption{ Secrecy outage probability for large $N_s$. Parameters: $\rho^{\left\{ p \right\}}_{ out}=0.1$, $\lambda_p = 10^{-4}\; m^{-2}$, $\lambda_s = 10^{-3}\; m^{-2}$, $\lambda_e = 10^{-4}\; m^{-2}$, $r_p=6\;m$, $\alpha =3$, $r_s=3\;m$, $R_s=1$, ${P_p}=15$~dBm, and  ${{\gamma _{th}^{\left\{ p \right\}} }}=6$~dBm.}
        \label{fig:8}
    \end{center}
\end{figure}

\section{Conclusion}
\label{conclusion}
In this paper, we considered  secure communication in  large scale  spectrum sharing networks in the presence of multiple non-colluding eavesdroppers.  We employed  beamforming and artificial noise generation (BF$\&$AN) at the SU transmitters to achieve secure transmission against  malicious eavesdroppers. { We obtained an exact  expression for the average secrecy rate, through which we observed the average secrecy rate boundary.} We also derived an exact expression for the  secrecy outage probability.  {Interestingly,  our results show that to achieve the optimal average secrecy rate,  more power should be allocated to AN with 
 increasing the density ratio between the eavesdroppers and the SUs; whereas for extremely low density ratio, all of the power should be allocated to information signal without injecting AN.}
Moreover, we derived the asymptotic average secrecy rate and the asymptotic secrecy outage probability
as the number of antennas at the SU transmitters grows large to showcase the large gain brought to the secrecy performance.

\appendices
\numberwithin{equation}{section}

\section{ A Proof of Lemma 1}\label{lemma1}

Consider a HPPP ${\Phi _s}$ with density $\lambda_s$, the aggregate interference from the SU  transmitters  is given by
\begin{align}
 {I_{s,z}} = \sum\nolimits_{i \in {\Phi _s}} {{W_{{s_i},z}}{{\left| {{X_{i,z}}} \right|}^{ - \alpha }}}.
 \end{align}
  The Laplace transform  of ${I_{s,z}}$ is
\begin{align}\label{lap_w1}
{\mathcal{L}_{{I_{s,z}}}}\left( s \right) =& \mathbbm{E}\Big( {\mathop \prod \limits_{i \in {\Phi _s}} {\mathbbm{E}_{{W_{{s_i},z}}}}\left( {\exp \left( { - s{W_{{s_i},z}}{{\left| {{X_{i,z}}} \right|}^{ - \alpha }}} \right)} \right)} \Big).
\end{align}
Applying  the Generating functional of HPPP in \cite{stoyanstochastic} and the polar-coordinate system, we have
\begin{align}\label{lap_w2}
 {\mathcal{L}_{{I_{s,z}}}}\left( s \right)
=& 
\exp \Big( { - {\lambda _s}\pi \mathbbm{E}\big[ {W_{{s_i},z}^{\frac{2}{\alpha }}} \big]\Gamma \big( {1 - \frac{2}{\alpha }} \big){s^{\frac{2}{\alpha }}}} \Big).
\end{align}

Then we turn our attention to derive  the expectation of  ${{W_{{s_i},z}}}$.
According to \cite{ngo2012uplink} and \cite{xi2013},  ${{{\Big| {{{\bf{h}}_{0,z}}\frac{{{\bf{h}}_{i,{s_i}}^\dag }}{{\left\| {{\bf{h}}_{{i},{s_i}}^\dag } \right\|}}} \Big|^2}}} \mathop  \sim Exp\left( {1} \right)$, and ${\left\| {{{\bf{h}}_{i,z}}{{\bf{G}}_{i,{s_i}}}} \right\|^2} \sim Gamma\left( {{N_s} - 1,1} \right)$.
Thus, we have the PDF distribution of   ${W_{{s_i},z}} = \sigma _s^2{\Big| {{{\bf{h}}_{i,z}}\frac{{{\bf{h}}_{i,{s_i}}^\dag }}{{\left\| {{\bf{h}}_{i,{s_i}}^\dag } \right\|}}} \Big|^2} + \sigma _n^2{\left\| {{{\bf{h}}_{i,z}}{{\bf{G}}_{i,{s_i}}}} \right\|^2}$  as
\begin{align}\label{PDF_Wz}
&{f_{{W_{{s_i},z}}}}\left( x \right)  =
  \left\{
\begin{array}{ll}
{\big( {1 - \frac{{{P_A}}}{{\left( {{N_s} - 1} \right){P_I}}}} \big)^{1 - {N_s}}}{\big( {{P_I}{e^{\frac{x}{{{P_I}}}}}} \big)^{ - 1}}\Big[ {1 - } \Big.\\  \hspace{0.6cm} \Big. {\sum\limits_{k = 0}^{{N_s} - 2} {{{\big( {\frac{{{N_s} - 1}}{{{P_A}}} - \frac{1}{{{P_I}}}} \big)}^k}\frac{{{x^k}}}{{k!}}{e^{ - \big( {\frac{{{N_s} - 1}}{{{P_A}}} - \frac{1}{{{P_I}}}} \big)x}}} } \Big]\\
 \hspace{5.04cm}\mbox{$\mu  \ne \frac{1}{{{N_s}}}$},\\
\frac{{{x^{{N_s} - 1}}{e^{ - \frac{x}{{{P_I}}}}}}}{{{P_I}^{{N_s} - 1}\left( {{N_s} - 1} \right)!}}
\hspace{3.0cm}\mbox{$\mu  = \frac{1}{{{N_s}}}$}.
\end{array}
\right.
\end{align}

  Taking the expectation of  ${{W_{{s_i},z}}}$ by using
\begin{align}
   \mathbbm{E}\big[ {W_{{s_i},z}^{\frac{2}{\alpha }}} \big] = \int_0^\infty  {x^{\frac{2}{\alpha }}{f_{{W_{{s_i},z}}}}\left( x \right)dx},
\end{align}   
    and substituting the derived expression of  $\mathbbm{E}\big[ {W_{{s_i},z}^{\frac{2}{\alpha }}} \big]$ into \eqref{lap_w2}, we  obtain \eqref{L_PDF_Wz}.

  \section{ A Proof of Theorem 1}\label{theorem1}

  According to the SIR of the typical PU receiver in \eqref{PU_SIR_AN},  we define  the sum interference at the typical PU receiver as 
{\begin{align}  
  {I_{Pr i, AN}} = {I_{p,{p_0}}} + P_p^{ - 1}{I_{s,{p_0}}},
  \end{align}}
   thus   the CDF of ${\gamma _{\tt SIR}^p}$  is expressed as
\begin{align}\label{CDFpri_1}
F_{\gamma _{\tt SIR}^p}^{\left\{ p \right\}}\big( {{\gamma _{th}^{\left\{ p \right\}} }} \big) = &{\mathbbm{E}_{{\Phi _p}}}\left\{ {{\mathbbm{E}_{{\Phi _s}}}\left\{ {\Pr \left\{ {\left. {{{\left| {{h_{{p_0}}}} \right|}^2} \le \gamma _{th}^{\left\{ p \right\}}{I_{Pri,AN}}{r_p}^\alpha } \right|} \right.} \right.} \right.
\nonumber \\ &\Big. {\Big. {\Big. {{\Phi _s},{\Phi _p}} \Big\}} \Big\}} \Big\}
  = 1 - {\mathcal{L}_{{I_{{Pr i, AN}}}}}\big( {\gamma _{th}^{\left\{ p \right\}}{r_p}^\alpha } \big)
\end{align}

By utilizing  similar approach in Appendix A of  \cite{zaidi2014} and based on  \textbf{Lemma 1}, we derive the outage probability at the typical PU receiver as
\begin{align}\label{CDFpriAN_2}
&P_{out}^{pri,AN}\big( {\gamma _{th}^{\left\{ p \right\}}} \big)  = \nonumber \\
&\left\{
\begin{array}{ll}
 1 - \exp \Big( { - \pi \big( {{\lambda _p}\Gamma \left( {1 + \frac{2}{\alpha }} \right) + {\lambda _s}{{\big( {\frac{{{P_s}}}{{{P_p}}}}  \big)}^{\frac{2}{\alpha }}}{\Upsilon _1}} \big)\delta} \Big)\\
 \hspace{6.7cm}\mbox{$\mu  \ne \frac{1}{{{N_s}}}$},\\
1 - \exp \Big( { - \pi \big( {{\lambda _p}\Gamma \left( {1 + \frac{2}{\alpha }} \right) + {\lambda _s}{{\big( {\mu {\frac{{{P_s}}}{{{P_p}}}} } \big)}^{\frac{2}{\alpha }}}\frac{{\Gamma \left( {{N_s} + \frac{2}{\alpha }} \right)}}{{\Gamma \left( {{N_s}} \right)}}} \big)\delta } \Big)\\
\hspace{6.7cm}\mbox{$\mu  = \frac{1}{{{N_s}}}$}.
\end{array}
\right.
\end{align}
where $\delta  = \Gamma \left( {1 - \frac{2}{\alpha }} \right){\big( {\gamma _{th}^{\left\{ p \right\}}} \big)^{\frac{2}{\alpha }}}{r_p}^2$.  By inversing \eqref{CDFpriAN_2}, we can derive the maximum permissive transmit power at the SU transmitters as \eqref{Ps_AN}.

\section{ A Proof of Lemma 2}\label{app_gsc_CDF}

The PDF and CDF of ${{\left\| {{\bold{h}_{0, {s_0}}}} \right\|}^2}$ are given by
\begin{align}\label{CDF1}
{f_{{{\left\| {{{\bf{h}}_{0, {s_0}}}} \right\|}^2}}}\left( x \right) = \frac{{{x^{{N_s} - 1}}{e^{ - x}}}}{{\left( {{{\rm{N}}_s} - 1} \right)!}},
\end{align}
and
\begin{align}\label{CDF2}
{F_{{{\left\| {{{\bf{h}}_{0, {s_0}}}} \right\|}^2}}}\left( x \right) = 1 - {e^{ - x}}\big( {\sum\limits_{m = 0}^{{N_s} - 1} {\frac{{{x^m}}}{{m!}}} } \big),
\end{align}
respectively.

Let us  define 
${I_{Sec,AN}} = {P_p}{I_{p,{s_0}}} + {I_{s,{s_0}}}$. 

 Based on the SIR  in \eqref{SU_SIR_AN},  the CDF of $\gamma _{\tt SIR}^{{s,AN}}$ can be represented as
\begin{align}\label{CDFsec_1}
&F_{\gamma _{\tt SIR}^{s,AN}}^{\left\{ s \right\}}\big( {\gamma _{th}^{\left\{ s \right\}}} \big) 
\nonumber \\
 &  =  1 - \sum\limits_{m = 0}^{{N_s} - 1} {{\mathbbm{E}_{{\Phi _p}}}} \Big\{ {{\mathbbm{E}_{{\Phi _s}}}\Big\{ {\int_0^\infty  {{e^{ - \tau \gamma _{th}^{\left\{ s \right\}}r_s^\alpha \sigma _s^{ - 2}}}{{\big( {\tau \gamma _{th}^{\left\{ s \right\}}r_s^\alpha \sigma _s^{ - 2}} \big)}^m}} } \big.} \Big.
   \nonumber \\& \hspace{0.4cm}\big. {\Big. {d\Pr \big( {{I_{Sec,AN}} \le \tau } \big)} \Big\}} \Big\}\frac{1}{{m!}}
\nonumber \\
&\mathop  = \limits^{(a)} 1 - {\mathbbm{E}_{{\Phi _p}}}\Big\{ {{\mathbbm{E}_{{\Phi _s}}}\Big\{ {\int_0^\infty  {{e^{ - \tau \gamma _{th}^{\left\{ s \right\}}r_s^\alpha \sigma _s^{ - 2}}}d\Pr \big( {{I_{Sec,AN}} \le \tau } \big)} } \Big\}} \Big\}
\nonumber \\
& \hspace{0.2cm}- \sum\limits_{m = 1}^{{N_s} - 1} {\frac{{{{\left( {r_s^\alpha } \right)}^m}}}{{m!{{\left( { - 1} \right)}^m}}}} {\mathbbm{E}_{{\Phi _p}}}\Big\{ {{\mathbbm{E}_{{\Phi _s}}}\Big\{ {\int_0^\infty  {{{\left. {\frac{{{d^m}\big( {{e^{ - \tau \gamma _{th}^{\left\{ s \right\}}x\sigma _s^{ - 2}}}} \big)}}{{d{x^m}}}} \right|}_{x = r_s^\alpha }}} } \Big.} \Big.
\nonumber \\
&\hspace{0.2cm}\Big. {\Big. {d\Pr \left( {{I_{Sec,AN}} \le \tau } \right)} \Big\}} \Big\},
\end{align}
where (a) follows from the fact that
\begin{align}
 {\left. {\frac{{{d^m}({e^{ - \tau \gamma _{th}^{\left\{ s \right\}}x\sigma _s^{ - 2}}})}}{{d{x^m}}}} \right|_{x = r_s^\alpha }} = {( - \tau \gamma _{th}^{\{ s\} }\sigma _s^{ - 2})^m}{e^{ - \tau \gamma _{th}^{\{ s\} }r_s^\alpha \sigma _s^{ - 2}}}. 
\end{align}

After some manipulations, we have
\begin{align}\label{CDFsec_2}
& F_{\gamma _{\tt SIR}^{s,AN}}^{\left\{ s \right\}}\big( {\gamma _{th}^{\left\{ s \right\}}} \big)
=1 - {{\cal L}_{{I_{Sec,AN}}}}\big( {\gamma _{th}^{\left\{ s \right\}}r_s^\alpha \sigma _s^{ - 2}} \big)
\nonumber \\
&\hspace{1cm} - \sum\limits_{m = 1}^{{N_s} - 1} {\frac{{{{\big( {r_s^\alpha } \big)}^m}}}{{m!{{\big( { - 1} \big)}^m}}}} {\big. {\frac{{{d^m}\big\{ {{{\cal L}_{{I_{Sec,AN}}}}\big( {\gamma _{th}^{\left\{ s \right\}}x\sigma _s^{ - 2}} \big)} \big\}}}{{d{x^m}}}} \big|_{x = r_s^\alpha }}.
\end{align}

We then need to derive the Laplace transform of ${I_{Sec,AN}}$.  Utilizing  \cite[eq. (4)]{zaidi2014} and \textbf{Lemma 2}, we obtain
\begin{align}\label{CDFsec_3}
&{{\cal L}_{{I_{Sec,AN}}}}\big( {\gamma _{th}^{\left\{ s \right\}}r_s^\alpha \sigma _s^{ - 2}} \big) = \exp \big( { - {\Lambda _l}{\big( {\gamma _{th}^{\left\{ s \right\}}} \big)}^{\frac{2}{\alpha }}r_s^2} \big),
\end{align}
where $\Lambda _l$ is given in \eqref{lambda}.

Now, we apply the Fa$\grave{a}$
 di Bruno's formula to solve the derivative of $m$th order as follows:
%
\begin{align}\label{deri1}
&{\big. {\frac{{{d^m}\big[ {\exp \big( { - {\Lambda _l}{{\big( {\gamma _{th}^{\left\{ s \right\}}} \big)}^{\frac{2}{\alpha }}}{x^{\frac{2}{\alpha }}}} \big)} \big]}}{{d{x^m}}}} \big|_{x = r_s^\alpha }}  = \exp \big( { - {\Lambda _l}{{\big( {\gamma _{th}^{\left\{ s \right\}}} \big)}^{\frac{2}{\alpha }}}r_s^2} \big)\nonumber \\
& \hspace{1cm}\sum m!\prod\limits_{j = 1}^m {\frac{{{{(( - {\Lambda _l}{{(\gamma _{th}^{\left\{ s \right\}})}^{\frac{2}{\alpha }}})\prod\limits_{k = 0}^{j - 1} {(\frac{2}{\alpha } - k)} {{({r_s})}^{2 - j\alpha }})}^{{m_j}}}}}{{{m_j}!j{!^{{m_j}}}}}}.
\end{align}

By substituting \eqref{deri1} into \eqref{CDFsec_2}, we get the closed-form expression for the CDF of SIR at the typical secondary user as \eqref{CDFsecan2}.

\section{ A Proof of Lemma 3}\label{app_gsc_CDF}

Let us define ${I_{Eve,AN}} = {P_p}{I_{p,{e_k}}} + {I_{s,{e_k}}} + \sigma _n^2{I_{{s_0},{e_k},an}}$. 

The CDF of ${\gamma _{\tt SIR}^e}$ can be written as
\begin{align}\label{CDFe_2}
\hspace{-0.3cm} F_{\gamma _{\tt SIR}^{e,AN}}^{\left\{ e \right\}}\big( {\gamma _{th}^{\{ e\} }} \big)
& = {\mathbbm{E}_{{\Phi _e}}}\Big\{ {{\mathbbm{E}_{{\Phi _p}}}\Big\{ {{\mathbbm{E}_{{\Phi _s}}}\Big\{ {{\prod _{{e_k} \in {\Phi _e}}}\Pr \big\{ {{{\big| {{{\bf{h}}_{{0},e}}\frac{{{\bf{h}}_{{0},{s_i}}^\dag }}{{\big\| {{\bf{h}}_{{0},{s_i}}^\dag } \big\|}}} \big|}^2} \le } \big.} \Big.} \Big.} \Big.
\nonumber \\
&\hspace{-1cm}\Big. {\Big. {\Big. {\Big. {\big. {\sigma _s^{ - 2}{I_{Eve,AN}}\gamma _{th}^{\{ e\} }{{\big| {{X_{e_k}}} \big|}^\alpha }} \big|{\Phi _s},{\Phi _p},{\Phi _e}} \Big\}} \Big\}} \Big\}} \Big\}.
\end{align}

According to \cite{ngo2012uplink},  $ {{{\bf{h}}_{{0},{e_k}}}\frac{{{\bf{h}}_{{0},{s_i}}^\dag }}{{\left\| {{\bf{h}}_{{0},{s_i}}^\dag } \right\|}}} $  is a zero-mean complex Gaussian variable, which is independent of ${{\bf{h}}_{{0},{s_i}}^\dag }$, and ${{{\big| {{{\bf{h}}_{{0},{e_k}}}\frac{{{\bf{h}}_{{0},{s_i}}^\dag }}{{\left\| {{\bf{h}}_{{0},{s_i}}^\dag } \right\|}}} \big|^2}}}$ follows the exponential distribution with unit mean. Thus, the CDF of ${\gamma _{\tt SIR}^e}$ can be represented as
\begin{align}\label{CDFe_22}
\hspace{-0.2cm} F_{\gamma _{\tt SIR}^{e,AN}}^{\left\{ e \right\}}\big( {\gamma _{th}^{\{ e\} }} \big)
& =  {\mathbbm{E}_{{\Phi _e}}}\Big\{ {{\prod _{{e_k} \in {\Phi _e}}}\Big( {1 - {\mathbbm{E}_{{\Phi _p}}}\Big\{ {{\mathbbm{E}_{{\Phi _s}}}\Big\{ {\int_0^\infty  {{e^{ - \tau \sigma _s^{ - 2}\gamma _{th}^{\{ e\} }}}} } \Big.} \Big.} \Big.} \Big.
\nonumber \\
&\hspace{1cm}\Big. {\Big. {\Big. {\Big. {{{\big| {{X_{e_k}}} \big|}^\alpha }d\Pr \big( {{I_{Eve,AN}} \le \tau } \big)} \Big\}} \Big\}} \Big)} \Big\}.
\end{align}

According to the proof of Lemma 3.1 in \cite{baccelli2006}, we express \eqref{CDFe_22} as
\begin{align}\label{CDFe_3}
F_{\gamma _{\tt SIR}^{e,AN}}^{\left\{ e \right\}}\big( {{\gamma _{th}}} \big) =& {\mathbbm{E}_{{\Phi _e}}}\Big\{ {{\prod _{{e_k} \in {\Phi _e}}}\Big( {1 - {{\cal L}_{{I_{Eve,AN}}}}\big( {\sigma _s^{ - 2}\gamma _{th}^{\{ e\} }{{\left| {{X_{e_k}}} \right|}^\alpha }} \big)} \Big)} \Big\}.
\end{align}

By using the  Generating functional of HPPP ${\Phi _e}$ \cite{stoyanstochastic}, we solve \eqref{CDFe_3} as
\begin{align}\label{CDFe_4}
&F_{\gamma _{\tt SIR}^{e,AN}}^{\left\{ e \right\}}\left( {{\gamma _{th}}} \right) \nonumber =  \exp \Big[ { - {\lambda _e}{\smallint _{{R^2}}}{{\cal L}_{{I_{Eve,AN}}}}\big( {\sigma _s^{ - 2}\gamma _{th}^{\{ e\} }{{\left| {{X_{e_k}}} \right|}^\alpha }} \big)de} \Big]
 \nonumber \\
 &= \exp \Big[ { - 2\pi {\lambda _e}\int_0^\infty  {{{\cal L}_{{I_{Eve,AN}}}}\big( {\sigma _s^{ - 2}\gamma _{th}^{\{ e\} }{{\left| {{X_{e_k}}} \right|}^\alpha }} \big)\left| {{X_{e_k}}} \right|d\left| {{X_{e_k}}} \right|} } \Big].
\end{align}

Now we utilize  \cite[eq. (4)]{zaidi2014} and \textbf{Lemma 2},   
we derive the Laplace transform of $I_{Eve,AN}$
as
\begin{align}\label{CDFe_54}
&{\mathcal{L}_{{I_{Eve,AN}}}}\left( s \right)  = \nonumber\\& \hspace{0.1cm}
\left\{
\begin{array}{ll}
 \exp \Big( { - \pi \big( {{\lambda _p}\Gamma \big( {1 + \frac{2}{\alpha }} \big){\eta ^{ - \frac{2}{\alpha }}}{\mu ^{ - \frac{2}{\alpha }}} + {\lambda _s}\frac{{\Gamma \big( {{N_s} + \frac{2}{\alpha }} \big)}}{{\Gamma \big( {{N_s}} \big)}}} \big)} \Big.\\
\hspace{0.4cm} \Big. {\Gamma \big( {1 - \frac{2}{\alpha }} \big){{\big( {\gamma _{th}^{\{ e\} }} \big)}^{\frac{2}{\alpha }}}{{\big| {{X_{e_k}}} \big|}^2}} \Big){\big( {\frac{{1 - \mu }}{{\big( {{N_s} - 1} \big)\mu }}\gamma _{th}^{\{ e\} } + 1} \big)^{ - \left( {{N_s} - 1} \right)}}\\
 \hspace{6.5cm}\mbox{$\mu  = \frac{1}{{{N_s}}}$},\\
\exp \Big( { - \pi } \Big.\big( {{\lambda _p}\Gamma \left( {1 + \frac{2}{\alpha }} \right){\eta ^{ - \frac{2}{\alpha }}} + {\lambda _s}{\Upsilon _1}} \big)\Gamma \left( {1 - \frac{2}{\alpha }} \right)\\
\hspace{0.4cm}\Big. {{{\big( {\gamma _{th}^{\{ e\} }} \big)}^{\frac{2}{\alpha }}}{\mu ^{ - \frac{2}{\alpha }}}{{\left| {{X_{e_k}}} \right|}^2}} \Big){\big( {\frac{{1 - \mu }}{{\left( {{N_s} - 1} \right)\mu }}\gamma _{th}^{\{ e\} } + 1} \big)^{ - \left( {{N_s} - 1} \right)}}\\
 \hspace{6.5cm}\mbox{$\mu \ne \frac{1}{{{N_s}}}$}.
\end{array}
\right.
\end{align}

By substituting \eqref{CDFe_54} into \eqref{CDFe_4}, we obtain
\begin{align}\label{CDFe_6}
F_{\gamma _{\tt SIR}^e}^{\left\{ e \right\}}\big( {\gamma _{th}^{\{ e\} }} \big) = &\exp \Big[ { - 2\pi {\lambda _e}{{\Big( {\frac{{1 - \mu }}{{\big( {{N_s} - 1} \big)\mu }}\gamma _{th}^{\{ e\} } + 1} \Big)}^{ - \left( {{N_s} - 1} \right)}}} \Big.
\nonumber \\
 &
\hspace{-1.5cm} \Big. {\int_0^\infty  {\exp \big( { - {\Lambda _l}{{\big( {\gamma _{th}^{\{ e\} }} \big)}^{\frac{2}{\alpha }}}{{\left| {{X_{e_k}}} \right|}^2}} \big)\left| {{X_{e_k}}} \right|d\left| {{X_{e_k}}} \right|} } \Big],
\end{align}
where $\Lambda _l$ is given in \eqref{lambda}.

By applying  \cite[Eq. 3.326.2.10]{gradshteyn00},
we derive the CDF of ${\gamma _{\tt SIR}^{e,AN}}$ as \eqref{CDFeveAN_2}.

\bibliographystyle{IEEEtran}
\bibliography{mybib}

\begin{thebibliography}{10}
\providecommand{\url}[1]{#1}
\csname url@samestyle\endcsname
\providecommand{\newblock}{\relax}
\providecommand{\bibinfo}[2]{#2}
\providecommand{\BIBentrySTDinterwordspacing}{\spaceskip=0pt\relax}
\providecommand{\BIBentryALTinterwordstretchfactor}{4}
\providecommand{\BIBentryALTinterwordspacing}{\spaceskip=\fontdimen2\font plus
\BIBentryALTinterwordstretchfactor\fontdimen3\font minus
  \fontdimen4\font\relax}
\providecommand{\BIBforeignlanguage}[2]{{%
\expandafter\ifx\csname l@#1\endcsname\relax
\typeout{** WARNING: IEEEtran.bst: No hyphenation pattern has been}%
\typeout{** loaded for the language `#1'. Using the pattern for}%
\typeout{** the default language instead.}%
\else
\language=\csname l@#1\endcsname
\fi
#2}}
\providecommand{\BIBdecl}{\relax}
\BIBdecl

\bibitem{yan2015ICC}
Y.~Deng, L.~Wang, S.~A.~R. Zaidi, J.~Yuan, and M.~Elkashlan, ``On the security
  of large scale spectrum sharing networks,'' in \emph{Proc. IEEE Int. Conf.
  Commun. (ICC)}, Jun. 2015, pp. 4877--4882.

\bibitem{goldsmith09}
A.~Goldsmith, S.~Jafar, I.~Maric, and S.~Srinivasa, ``Breaking spectrum
  gridlock with cognitive radios: An information theoretic perspective,''
  \emph{Proc. IEEE}, vol.~97, no.~5, pp. 894--914, May. 2009.

\bibitem{yan2015}
Y.~Deng, L.~Wang, M.~Elkashlan, K.~J. Kim, and T.~Q. Duong, ``Generalized
  selection combining for cognitive relay networks over nakagami- m fading,''
  \emph{{IEEE} Trans. Signal Process.}, vol.~63, no.~8, pp. 1993--2006, Apr.
  2015.

\bibitem{frag2013survey}
A.~G. Fragkiadakis, E.~Z. Tragos, and I.~G. Askoxylakis, ``A survey on security
  threats and detection techniques in cognitive radio networks,'' \emph{{IEEE}
  J. Sel. Areas Commun.}, vol.~15, no.~1, pp. 428--445, Jan. 2013.

\bibitem{mukherjee2014principle}
A.~Mukherjee, S.~A.~A. Fakoorian, J.~Huang, and A.~L. Swindlehurst,
  ``Principles of physical layer security in multiuser wireless networks: A
  survey,'' \emph{IEEE Commun. Surveys Tuts.}, vol.~16, no.~3, pp. 1550--1573,
  Mar. 2014.

\bibitem{wang2011anti}
B.~Wang, Y.~Wu, K.~J.~R. Liu, and T.~C. Clancy, ``An anti-jamming stochastic
  game for cognitive radio networks,'' \emph{{IEEE} J. Sel. Areas Commun.},
  vol.~29, no.~4, pp. 877--889, Apr. 2011.

\bibitem{zhou2011throughput}
X.~Zhou, R.~K. Ganti, J.~G. Andrews, and A.~Hjorungnes, ``On the throughput
  cost of physical layer security in decentralized wireless networks,''
  \emph{{IEEE} Trans. Wireless Commun.}, vol.~10, no.~8, pp. 2764--2775, Aug.
  2011.

\bibitem{wyner1975wire}
A.~D. Wyner, ``The wire-tap channel,'' \emph{Bell System Technical Journal},
  vol.~54, no.~8, pp. 1355--1387, Oct. 1975.

\bibitem{yan2016sensor}
Y.~Deng, L.~Wang, M.~Elkashlan, A.~Nallanathan, and R.~K. Mallik, ``Physical
  layer security in three-tier wireless sensor networks: {A} stochastic
  geometry approach,'' \emph{IEEE Trans. Inf. Forensics Security}, to appear
  2016.

\bibitem{tong2014transmission}
T.-X. Zheng, H.-M. Wang, and Q.~Yin, ``On transmission secrecy outage of a
  multi-antenna system with randomly located eavesdroppers,'' \emph{{IEEE}
  Commun. Lett.}, vol.~18, no.~8, pp. 1299--1302, Aug. 2014.

\bibitem{zheng2015multi}
T.~Zheng, H.~Wang, J.~Yuan, D.~Towsley, and M.~Lee, ``Multi-antenna
  transmission with artificial noise against randomly distributed
  eavesdroppers,'' \emph{IEEE Trans. Commun.}, to appear 2015.

\bibitem{chao2015uncoordinated}
C.~Wang, H.-M. Wang, X.~gen Xia, and C.~Liu, ``Uncoordinated jammer selection
  for securing {SIMOME} wiretap channels: A stochastic geometry approach,''
  \emph{{IEEE} Trans. Wireless Commun.}, vol.~14, no.~5, pp. 2596--2612, May.
  2015.

\bibitem{he2013physical}
H.~Wang, X.~Zhou, and M.~C. Reed, ``Physical layer security in cellular
  networks: A stochastic geometry approach,'' \emph{{IEEE} Trans. Wireless
  Commun.}, vol.~12, no.~6, pp. 2776--2787, Jun. 2013.

\bibitem{Geraci2014}
G.~Geraci, H.~S. Dhillon, J.~G. Andrews, J.~Yuan, and I.~B. Collings,
  ``Physical layer security in downlink multi-antenna cellular networks,''
  \emph{{IEEE} Trans. Commun.}, vol.~62, no.~6, pp. 2006--2021, Jun. 2014.

\bibitem{pinto2012secure1}
P.~C. Pinto, J.~Barros, and M.~Z. Win, ``Secure communication in stochastic
  wireless networks—part {I}: connectivity,'' \emph{IEEE Trans. Inf.
  Forensics Security}, vol.~7, no.~1, pp. 125--138, Feb. 2012.

\bibitem{elsawy2013stochastic}
H.~ElSawy, E.~Hossain, and M.~Haenggi, ``Stochastic geometry for modeling,
  analysis, and design of multi-tier and cognitive cellular wireless networks:
  A survey,'' \emph{IEEE Commun. Surveys and Tutorials}, vol.~15, no.~3, pp.
  996--1019, Jul. 2013.

\bibitem{huiming2015}
H.-M. Wang and X.-G. Xia, ``Enhancing wireless secrecy via cooperation: signal
  design and optimization,'' \emph{{IEEE} Commun. Mag.}, vol.~53, no.~12, pp.
  47--53, Dec. 2015.

\bibitem{shafiee2007achievable}
S.~Shafiee and S.~Ulukus, ``Achievable rates in {G}aussian {MISO} channels with
  secrecy constraints,'' in \emph{Proc. IEEE ISIT}, Nice, France, Jun. 2007,
  pp. 2466--2470.

\bibitem{goel2008guaranteeing}
S.~Goel and R.~Negi, ``Guaranteeing secrecy using artificial noise,''
  \emph{{IEEE} Trans. Wireless Commun.}, vol.~7, no.~6, pp. 2180--2189, Jun.
  2008.

\bibitem{xiang2010secure}
X.~Zhou and M.~R. McKay, ``Secure transmission with artificial noise over
  fading channels: Achievable rate and optimal power allocation,'' \emph{{IEEE}
  Trans. Veh. Technol.}, vol.~59, no.~8, pp. 3831--3842, Oct. 2010.

\bibitem{xi2013}
X.~Zhang, X.~Zhou, and M.~R. McKay, ``Enhancing secrecy with multi-antenna
  transmission in wireless ad hoc networks,'' \emph{{IEEE} Trans. Inf.
  Forensics Security}, vol.~8, no.~11, pp. 1802--1814, Nov. 2013.

\bibitem{swundke2009fixed}
A.~L. Swindlehurst, ``Fixed {SINR} solutions for the {MIMO} wiretap channel,''
  in \emph{Proc. IEEE Int. Conf. Acoust. Speech Signal Process.}, Taipei,
  China, Apr. 2009, pp. 2437--2440.

\bibitem{stanojev2013improve}
I.~Stanojev and A.~Yener, ``Improving secrecy rate via spectrum leasing for
  friendly jamming,'' \emph{{IEEE} Trans. Wireless Commun.}, vol.~12, no.~1,
  pp. 134--145, Jan. 2013.

\bibitem{ning2013cooperative}
N.~Zhang, N.~Lu, N.~Cheng, J.~W. Mark, and X.~S. Shen, ``Cooperative spectrum
  access towards secure information transfer for {CRN}s,'' \emph{{IEEE} J. Sel.
  Areas Commun.}, vol.~31, no.~11, pp. 2453--2464, Nov. 2013.

\bibitem{yuan2014secrecy}
Y.~He, J.~Evans, and S.~Dey, ``Secrecy rate maximization for cooperative
  overlay cognitive radio networks with artificial noise,'' in \emph{Proc. IEEE
  Int. Conf. Commun. (ICC)}, Jun. 2014, pp. 1663--1668.

\bibitem{yulong2013pysical}
Y.~Zou, X.~Wang, and W.~Shen, ``Physical-layer security with multiuser
  scheduling in cognitive radio networks,'' \emph{{IEEE} Trans. Commun.},
  vol.~61, no.~12, pp. 5103--5113, Dec. 2013.

\bibitem{yu2014secrecy}
Y.~Zou, X.~Li, and Y.~chang Liang, ``Secrecy outage and diversity analysis of
  cognitive radio systems,'' \emph{{IEEE} J. Sel. Areas Commun.}, vol.~32,
  no.~11, pp. 2222--2236, Nov. 2014.

\bibitem{yiyang2013secure}
Y.~Pei, Y.-C. Liang, L.~Zhang, K.~C. Teh, and K.~H. Li, ``Secure communication
  over {MISO} cognitive radio channels,'' \emph{{IEEE} Trans. Wireless
  Commun.}, vol.~9, no.~4, pp. 1494--1502, Apr. 2010.

\bibitem{yiyang11secure}
Y.~Pei, Y.-C. Liang, K.~C. Teh, and K.~H. Li, ``Secure communication in
  multiantenna cognitive radio networks with imperfect channel state
  information,'' \emph{{IEEE} Trans. Signal Process.}, vol.~59, no.~4, pp.
  1683--1693, Apr. 2011.

\bibitem{chao2014secrecy}
C.~Wang and H.-M. Wang, ``On the secrecy throughput maximization for {MISO}
  cognitive radio network in slow fading channels,'' \emph{IEEE Trans. Inf.
  Forensics Security}, vol.~9, no.~11, pp. 1814--1827, Nov. 2014.

\bibitem{anand2008cognitive}
S.~Anand and R.~Chandramouli, ``On the secrecy capacity of fading cognitive
  wireless networks,'' in \emph{Proc. IEEE Int. Conf. Cogn. Radio Oriented
  Wireless Netw. Commun.}, May 2008, pp. 1--5.

\bibitem{shu2013physical}
Z.~Shu, Y.~Qian, and S.~Ci, ``On physical layer security for cognitive radio
  networks,'' \emph{IEEE Netw.}, vol.~27, no.~3, pp. 28--33, Jun. 2013.

\bibitem{weber2010overview}
S.~Weber, J.~G. Andrews, and N.~Jindal, ``An overview of the transmission
  capacity of wireless networks,'' \emph{{IEEE} Trans. Commun.}, vol.~58,
  no.~12, pp. 3593--3604, Dec. 2010.

\bibitem{zaidi2014}
S.~Zaidi, D.~McLernon, and M.~Ghogho, ``Breaking the area spectral efficiency
  wall in cognitive underlay networks,'' \emph{{IEEE} J. Sel. Areas Commun.},
  vol.~32, no.~11, pp. 1--17, Feb. 2014.

\bibitem{baccelli2006}
F.~Baccelli, B.~Blaszczyszyn, and P.~Muhlethaler, ``An aloha protocol for
  multihop mobile wireless networks,'' \emph{{IEEE} Trans. Inf. Theory},
  vol.~52, no.~2, pp. 421--436, Feb. 2006.

\bibitem{stoyanstochastic}
D.~Stoyan, W.~Kendall, and J.~Mecke, ``Stochastic geometry and its
  applications,'' \emph{Wiley New York}, vol.~2, 1987.

\bibitem{LunDong}
L.~Dong, Z.~Han, A.~P. Petropulu, and H.~V. Poor, ``Improving wireless physical
  layer security via cooperating relays,'' \emph{{IEEE} Trans. Signal
  Process.}, vol.~58, no.~3, pp. 1875--1888, Mar. 2010.

\bibitem{yan2014GSC}
Y.~Deng, M.~Elkashlan, P.~L. Yeoh, N.~Yang, and R.~K. Mallik, ``Cognitive
  {MIMO} relay networks with generalized selection combining,'' \emph{{IEEE}
  Trans. Wireless Commun.}, vol.~13, no.~9, pp. 4911--4922, Sep. 2014.

\bibitem{yan2015GSC}
Y.~Deng, M.~Elkashlan, N.~Yang, P.~L. Yeoh, and R.~K. Mallik, ``Impact of
  primary network on secondary network with generalized selection combining,''
  \emph{{IEEE} Trans. Veh. Technol.}, vol.~64, no.~7, pp. 3280--3285, Jul.
  2015.

\bibitem{Xiangyun2010}
X.~Zhou and M.~R. McKay, ``Secure transmission with artificial noise over
  fading channels: Achievable rate and optimal power allocation,'' \emph{{IEEE}
  Trans. Veh. Technol.}, vol.~59, no.~8, pp. 3831--3842, Oct. 2010.

\bibitem{lifeng2014physical}
L.~Wang, N.~Yang, M.~Elkashlan, P.~L. Yeoh, and J.~Yuan, ``Physical layer
  security of maximal ratio combining in two-wave with diffuse power fading
  channels,'' \emph{IEEE Trans. Inf. Forensics Security}, vol.~9, no.~2, pp.
  247--258, Feb. 2014.

\bibitem{hong2013enhancing}
Y.-W.~P. Hong, P.-C. Lan, and C.-C.~J. Kuo, ``Enhancing physical-layer secrecy
  in multiantenna wireless systems: An overview of signal processing
  approaches,'' \emph{IEEE Signal Processing Mag.}, vol.~30, no.~5, pp. 29--40,
  Sep. 2013.

\bibitem{haenggi2009stochastic}
M.~Haenggi, J.~G. Andrews, F.~Baccelli, O.~Dousse, and M.~Franceschetti,
  ``Stochastic geometry and random graphs for the analysis and design of
  wireless networks,'' \emph{{IEEE} J. Sel. Areas Commun.}, vol.~27, no.~7, pp.
  1029--1046, Sep. 2009.

\bibitem{Hollenbeck1998}
K. J. Hollenbeck, Invlap.m: A Matlab Function for Numerical Inversion of
  Laplace Transforms by the de Hoog Algorithm 1998. [Online].
  Available:http://www.mathworks.com/matlabcentral/fileexchange/
  32824-numerical-inversion-of-Laplace-transforms-in-matlab/content/ INVLAP.m.

\bibitem{venkataraman2006shot}
J.~Venkataraman, M.~Haenggi, and O.~Collins, ``Shot noise models for outage and
  throughput analyses in wireless ad hoc networks,'' in \emph{Proc. 44th Annu.
  Allerton Conf. Communication, Control, and Computing, Monticello, IL}, Sep.
  2006, pp. 1--7.

\bibitem{wendel1961}
J.~G. Wendel, ``The non-absolute convergence of {G}il-{P}elaez' inversion
  integral,'' \emph{The Annals of Mathematical Statistics}, vol.~32, no.~1, pp.
  338--339, Mar. 1961.

\bibitem{ngo2012uplink}
H.~Q. Ngo, M.~Matthaiou, T.~Q. Duong, and E.~G. Larsson, ``Uplink performance
  analysis of multicell {MU-MIMO} systems with {ZF} receivers,'' \emph{{IEEE}
  Trans. Veh. Technol.}, vol.~62, no.~9, pp. 4471--4482, Nov. 2013.

\bibitem{gradshteyn00}
I.~S. Gradshteyn and I.~M. Ryzhik, \emph{Table of Integrals, Series and
  Products}, 6th~ed.\hskip 1em plus 0.5em minus 0.4em\relax New York, NY, USA:
  Academic Press, 2000.

\end{thebibliography}

\balance
\end{document}